\def\withcomments{1}
\def\full{1}
\documentclass[11pt,english]{article}
\usepackage{babel,comment}
\usepackage{amsmath, amsthm, amssymb, url}
\usepackage[numbers,longnamesfirst]{natbib}
\usepackage{verbatim}
\usepackage{cleveref,paralist}
\usepackage[usenames,dvipsnames]{color} 
\usepackage{multirow}
\usepackage{bigstrut}
\usepackage[disable]{todonotes}
\usepackage{fullpage, bbm}

\usepackage[ruled, noend, noline]{algorithm2e}

\ifnum\full=0
\usepackage[compact]{titlesec}
\titlespacing{\section}{0pt}{2ex}{1ex}
\titlespacing{\subsection}{0pt}{1ex}{1ex}
\titlespacing{\subsubsection}{0pt}{0.5ex}{0ex}
\renewenvironment{itemize}[1]{\begin{compactitem}#1}{\end{compactitem}}
\renewenvironment{enumerate}[1]{\begin{compactenum}#1}{\end{compactenum}}
\fi

\newtheorem{theorem}{Theorem}[section]
\newtheorem{lemma}[theorem]{Lemma}

\newtheorem{proposition}[theorem]{Proposition}
\newtheorem{claim}[theorem]{Claim}

\newtheorem{corollary}[theorem]{Corollary}

\newtheorem{remark}[theorem]{Remark}
\newtheorem{definition}{Definition}[section]

\numberwithin{figure}{section}

\newcommand{\cP}{{\mathcal P}}

\newcommand{\C}{{\mathcal C}}
\newcommand{\eps}{{\epsilon}}

\newcommand{\dis}{dist}
\newcommand{\Dis}{Dist}

\newcommand{\ste}{{\epsilon\emph{-tester}}}

\newcommand{\Accept}{\textbf{Accept}\xspace}
\newcommand{\accept}{\textbf{accept}\xspace}

\newcommand{\reject}{\textbf{reject}\xspace}

\newcommand{\integerset}[1]{[0..{#1})}
\newcommand{\domain}{\integerset{n}^2}
\newcommand{\gp}{\text{GP}}
\newcommand{\side}{r}
\newcommand{\rblock}{$\side$-block\xspace}

\newcommand{\lind}{t} 

\newcommand{\hull}{\text{Hull}}

\newcommand{\Best}{{\sf Best}\xspace}

\newcommand{\Compst}{{\sf ComputeStatus}\xspace}
\newcommand{\Constgr}{{\sf ConstructGraph}\xspace}

\newcommand{\BestFixed}{{\sf Best For Fixed Base}\xspace}
\newcommand{\errle}{\dout_{\text{left}}}
\newcommand{\errri}{\dout_{\text{right}}}
\newcommand{\Tend}{{\bf T}_{\text{end}}}
\newcommand{\Tfin}{{\bf T}_{\text{fin}}}
\newcommand{\Tcut}{{\bf T}_{\text{cut}}}

\newcommand{\Tstart}{{\bf T}_0}
\newcommand{\stripset}{{\bf R}}

\newcommand{\err}{\mbox{\it err}_S}
\newcommand{\con}{144}

\newcommand{\hp}[2]{H^{#1}_{#2}}
\newcommand{\sepline}[2]{L^{#1}_{#2}}
\newcommand{\hpi}[2]{M^{#1}_{#2}}
\newcommand{\myerr}[1]{Err(#1)}

\DeclareMathOperator*{\E}{\mathbb E}
\DeclareMathOperator*{\Var}{\mathrm Var}

\newcommand{\mydelta}{\epsilon} 
\newcommand{\bigdelta}{{\epsilon_0}} 
\newcommand{\dsquares}{d_{\rm squares}}
\newcommand{\dhatsquares}{\hat{d}_{\rm squares}}
\newcommand{\dout}{\hat{d}}

\newenvironment{myproof}{\begin{proof}
}{
\end{proof}}

\ifnum\withcomments=1
   \newcommand{\mnote}[1]{{\color{red}\footnote{{\color{red} {\bf M:} #1}}}}
   \newcommand{\pnote}[1]{{\color{green}\footnote{{\color{green} {\bf P:} #1}}}}
   \newcommand{\snote}[1]{{\color{black}\footnote{{\color{black} {\bf S:} #1}}}}
\else
   \newcommand{\mnote}[1]{}
   \newcommand{\pnote}[1]{}
   \newcommand{\snote}[1]{}
\fi

\newcommand{\mch}[1]{{\color{black}#1}}

\date{}
\ifnum\full=1
\title{Tolerant Testers of Image Properties\footnote{A preliminary version of this article was published in the proceedings of the 43rd International Colloquium on Automata, Languages, and Programming, ICALP, 2016~\cite{BMR16icalp}}}
\else
\title{Tolerant Testers of Image Properties}
\fi
\author{Piotr Berman\thanks{Pennsylvania State University, USA; {\tt berman@cse.psu.edu}.}\\
\and Meiram Murzabulatov\thanks{Pennsylvania State University, USA; {\tt meyram85@yahoo.co.uk}. This author was supported by NSF CAREER award CCF-0845701 and NSF award CCF-1422975.
}\\
\and Sofya Raskhodnikova\thanks{Pennsylvania State University, USA; {\tt sofya@cse.psu.edu}. This author was supported by NSF CAREER award CCF-0845701, NSF award CCF-1422975, and by Boston University's Hariri Institute for Computing and Center for Reliable Information Systems and Cyber Security and, while visiting the Harvard Center for Research on Computation \& Society, by a Simons Investigator grant to Salil Vadhan.
}
}
\begin{document}
\raggedbottom
\setlength{\parskip}{0pt}
\maketitle
\begin{abstract}
We initiate a systematic study of tolerant testers of image properties or, equivalently, algorithms that approximate the distance from a given image to the desired property (that is, the smallest fraction of pixels that need to change in the image to ensure that the image satisfies the desired property). Image processing is a particularly compelling area of applications for sublinear-time algorithms and, specifically, property testing. However, for testing algorithms to reach their full potential in image processing, they have to be tolerant, which allows them to be resilient to noise. Prior to this work, only one tolerant testing algorithm for an image property (image partitioning) has been published.

We design efficient approximation algorithms for the following fundamental questions: What fraction of pixels have to be changed in an image so that it becomes a half-plane? a representation of a convex object? a representation of a connected object? More precisely, our algorithms approximate the distance to three basic properties (being a half-plane, convexity, and connectedness) within a small additive error $\eps$, after reading a number of pixels polynomial in $1/\eps$ and independent of the size of the image. The running time of the testers for half-plane and convexity is also polynomial in $1/\eps$.
Tolerant testers for these three properties were not investigated previously. For convexity and connectedness, even the existence of distance approximation algorithms with query complexity independent of the input size is not implied by previous work. (It does not follow from the VC-dimension bounds, since VC dimension of convexity and connectedness, even in two dimensions, depends on the input size. It also does not follow from the existence of non-tolerant testers.)

Our algorithms require very simple access to the input: uniform random samples for the half-plane property and convexity, and samples from uniformly random blocks for connectedness. However, the analysis of the algorithms, especially for convexity, requires many geometric and combinatorial insights. For example, in the analysis of the algorithm for convexity, we define a set of reference polygons  $P_\mydelta$ such that (1) every convex image has a nearby polygon in $P_\mydelta$ and (2) one can use dynamic programming to quickly compute the smallest empirical distance to a polygon in $P_\mydelta$. This construction might be of independent interest.

\end{abstract}
\section{Introduction}\label{sec:intro}
Image processing is a particularly compelling area of applications for sublinear-time algorithms and, specifically, property testing. Images are huge objects, and our visual system manages to process them very quickly without examining every part of the image. Moreover, many applications in image analysis have to process a large number of images online, looking for an image that satisfies a certain property among images that are generally very far from satisfying it. Or, alternatively, they look for a subimage satisfying a certain property in a large image (e.g., a face in an image where most regions are part of the background.) There is a growing number of proposed {\em rejection-based} algorithms that employ a quick test that is likely to reject a large number of unsuitable images (see, e.g., citations in \cite{KleinerKNB11}).

Property testing~\cite{RS96,GGR98} is a formal study of fast algorithms that accept objects with a given property and reject objects that are far. Testing image properties in this framework was first considered in \cite{Ras03}. Ron and Tsur~\cite{TsurR10} initiated property testing of images with a different input representation, suitable for testing properties of sparse images. Since these models were proposed, several sublinear-time algorithms for visual properties were implemented and used: namely, those by Kleiner et al.\ and Korman et al.\ \cite{KleinerKNB11,KormanRT,KormanRTA13}.

However, for sublinear-time algorithms to reach their full potential in image processing, they have to be resilient to noise: images are often noisy, and it is undesirable to reject images that differ only on a small fraction of pixels from an image satisfying the desired property. Tolerant testing was introduced by Parnas, Ron and Rubinfeld~\cite{PRR06}  exactly with this goal in mind---to deal with noisy objects. It builds on the property testing model and calls for algorithms that accept objects that are close to having a desired property and reject objects that are far. Another related task is approximating distance of a given object to a nearest object with the property within additive error~$\mydelta$. (Distance approximation algorithms imply tolerant testers in a straightforward
\ifnum\full=1
way: see the remark after Definition~\ref{def:distance-approx-alg}).
\else
way.)
\fi
The only image problem for which tolerant testers were studied is the image partitioning problem investigated by Kleiner et al.~\cite{KleinerKNB11}.
\subparagraph{Our results.}
We design efficient approximation algorithms for the following fundamental questions: What fraction of pixels have to be changed in an image so that it becomes a half-plane? a representation of a convex object? a representation of a connected object?
In other words, we design algorithms that approximate the distance to being a half-plane, convexity and connectedness within a small additive error or, equivalently, tolerant testers for these properties. These problems were not investigated previously in the tolerant testing framework.
For all three properties, we give $\mydelta$-additive distance approximation algorithms that run in constant time (i.e., dependent only on $\mydelta$, but not the size of the image). We remark that even though it was known that these properties can be tested in constant time \cite{Ras03}, this fact does not necessarily imply constant-query tolerant testers for these properties. E.g., Fischer and Fortnow~\cite{FischerF06} exhibit a property (of objects representable with strings of length $n$) which is testable with a constant number of queries, but for which every tolerant tester requires $n^{\Omega(1)}$ queries.
For convexity and connectedness, even the existence of distance approximation algorithms with query (or time) complexity independent of the input size does not follow from previous work. It does not follow from the VC-dimension bounds, since VC dimension of convexity and connectedness, even in two dimensions, depends on the input size\footnote{For $n\times n$ images, the VC dimension of convexity is $\Theta(n^{2/3})$ (this is the maximum number of vertices of a convex lattice polygon in an $n\times n$ lattice \cite{Barany00}); for connectedness, it is $\Theta(n)$.}. Implications of the VC dimension bound on convexity are further discussed below.

Our results on distance approximation are summarized  in Table~\ref{table:all-results}. Our algorithm for convexity is the most important and technically difficult of our results, requiring a large number of new ideas to get running time polynomial in $1/\eps.$ To achieve this, we define a set of reference polygons  $P_\mydelta$ such that (1) every convex image has a nearby polygon in $P_\mydelta$ and (2) one can use dynamic programming to quickly compute the smallest empirical distance to a polygon in $P_\mydelta$. It turns out that the empirical error of our algorithm is proportional to the sum of the square roots of the areas of the regions it considers in the dynamic program. To guarantee (2) and keep our empirical error small, our construction ensures that the sum of the square roots of the areas of the considered regions is small.
This construction might be of independent interest.
\begin{table*}[t]
\begin{center}
\begin{tabular}{| l || c | c| c |}
\hline
\multicolumn{1}{|c||}{Property}& Sample Complexity &  Run Time & Access to Input
\\
\hline
Half-plane  &  $O\left(\frac 1 {\mydelta^2} \log \frac 1 \mydelta\right)$ & $O\left(\frac 1 {\mydelta^3} \log \frac 1 \mydelta\right)$ & uniformly random pixels
\\
Convexity &  $O\left(\frac 1 {\mydelta^2} \log \frac 1 \mydelta\right)$ & $O\left(\frac 1 {\mydelta^8} \right)$ & uniformly random pixels
\\
Connectedness &  $O\left(\frac 1 {\mydelta^4} \right)$ \ \ \ \ & $\exp\left(O\left(\frac 1 \mydelta \right)\right)$ &uniformly random blocks of pixels
\\
\hline
\end{tabular}
\end{center}
\caption{Our results on distance approximation.
To get complexity of $(\eps_1,\eps_2)$-tolerant testing, substitute $\mydelta=(\eps_2-\eps_1)/2$.
}
\label{table:all-results}
\end{table*}
Our algorithms do not need sophisticated access to the input image: uniformly randomly sampled pixels suffice for our algorithms for the half-plane property and convexity. For connectedness, we allow our algorithms to query pixels from a uniformly random block. (See the end of Section~\ref{sec:defintions_notation} for a formal specification of the input access.)

Our algorithms for convexity and half-plane work by first implicitly learning the object\footnote{\label{fn:connection-to-learning}There is a known implication from learning to testing. As proved in \cite{GGR98}, a proper PAC learning algorithm for property ${\cal P}$
with sampling complexity $q(\eps)$ implies a 2-sided error (uniform) property tester for $\cal{P}$ that takes $q(\eps/2) + O(1/\eps)$ samples. There is an analogous implication from proper agnostic PAC learning to distance approximation with an overhead of $O(1/\eps^2)$ instead of $O(1/\eps)$. We choose to present our testers first and get learners as corollary because our focus is on testing and because we want additional features for our testers, such as 1-sided error, that do not automatically follow from the generic relationship.}.
\ifnum\full=1
PAC learning was defined by Valiant~\cite{Valiant84}, and agnostic learning, by Kearns et al.~\cite{KearnsSS94} and Haussler~\cite{Haussler92}.
\fi
 As a corollary of our analysis, we obtain fast proper agnostic PAC learners of half-planes and of convex sets in two dimensions that work under the uniform distribution.
 The sample and time complexity\footnote{All our results are stated for error probability $\delta=1/3$. To get results for general $\delta$, by standard arguments, it is enough to multiply the complexity of an algorithm by $\log 1/\delta$.} of 
 the PAC learners is as indicated in Table~\ref{table:all-results} for distance approximation algorithms for corresponding properties.

While the sample complexity of our agnostic half-plane learner (and hence our distance approximation algorithm for half-planes) follows from the VC dimension bounds, its running time does not. Agnostically learning half-spaces under the uniform distribution has been studied by \cite{KalaiKMS08}, but only for the hypercube $\{-1,1\}^d$ domains, not the plane.
Our PAC learner of convex sets, in contrast to our half-plane learner, \ dimension lower bounds on sample complexity. (The sample complexity of a PAC learner for a class
is at least proportional to the VC dimension of that class \cite{EHKV89}.) Since VC dimension of convexity of $n\times n$ images is
$\Theta(n^{2/3})$,
proper PAC learners of convex sets in two dimensions (that work under arbitrary distributions) must have sample complexity $\Omega(n^{2/3})$. However, one can do much better with respect to the uniform distribution.
Schmeltz~\cite{Sch92} showed that a non-agnostic learner for that task needs
$\Theta(\mydelta^{-3/2})$ samples.
Surprisingly, it appears that this question has not been studied at all for agnostic learners. Our agnostic learner for convex sets in
\ifnum\full=1
two dimensions
\else
2D
\fi
under the uniform distribution needs
$O\left(\frac 1 {\mydelta^2} \log \frac 1 \mydelta\right)$ samples and runs in time $O\left(\frac 1 {\mydelta^8} \right)$.

Finally, we note that for connectedness, we take a different approach. Our algorithms do not try to learn the object first; instead they rely on a combinatorial characterization of distance to connectedness. We show that distance to connectedness can be represented as an average of distances of sub-images to a related property.
\subparagraph{Comparison to other related work.}
Property testing has rich literature on graphs and functions, however, properties of images have been investigated very little. Even though superficially the inputs to various types of testing tasks might look similar, the problems that arise are different.
In the line of work on testing dense graphs, started by Goldreich et al.~\cite{GGR98}, the input is also an $n\times n$ binary matrix, but it represents an adjacency matrix of the dense input graph. So, the problems considered are different than in this work.
In the line of work on testing geometric properties, started by Czumaj, Sohler, and Ziegler~\cite{CzumajSZ00} and Czumaj and Sohler~\cite{CS01}, the input is a set of points represented by their coordinates. The allowed queries and the distance measure on the input space are different from ours. 

A line of work potentially relevant for understanding connectedness of images is on connectedness of bounded-degree graphs. Goldreich and Ron~\cite{GR02}  gave a tester for this property, subsequently improved by Berman et al.~\cite{BermanRY14}. Campagna et al.~\cite{CampagnaGR13} gave a tolerant tester for this problem. Even though we view our image as a graph in order to define connectedness of images, there is a significant difference in how distances between instances are measured (see~\cite{Ras03} for details). We also note, that unlike in~\cite{CampagnaGR13}, our tolerant tester for connectedness is fully tolerant, i.e., it works for all settings of parameters.

The only previously known tolerant tester for image properties was given by Kleiner et al.~\cite{KleinerKNB11}. They consider the following class of image partitioning problems, each specified by a $k\times k$ binary template matrix $T$ for a small {\em constant} $k$. The image satisfies the property corresponding to $T$ if it can be partitioned by $k-1$ horizontal and $k-1$ vertical lines into blocks, where each block has the same color as the corresponding entry of $T$. Kleiner et al.\ prove that $O(1/\mydelta^2)$ samples suffice for tolerant testing of image partitioning properties. Note that VC dimension of such a property is $O(1)$, so by Footnote~\ref{fn:connection-to-learning}, we can get a $O(1/\mydelta^2\log 1/\mydelta)$ bound. Our algorithms required numerous new ideas to significantly beat VC dimension bounds (for convexity and connectedness) and to get low running time.

For the properties we study, distance approximation algorithms and tolerant testers were not investigated previously. In the standard property testing model, the half-plane property can be tested in $O(\eps^{-1})$ time~\cite{Ras03}, convexity can be tested in $O(\eps^{-4/3})$ time~\cite{BMR16socg}, and connectedness can be tested in $O(\eps^{-2}\log \eps^{-1})$ time~\cite{Ras03,BermanRY14}. As we explained, property testers with running time independent of $\eps$ do not necessarily imply tolerant testers with that feature.
Many new ideas are needed to obtain our tolerant testers. In particular, the standard testers for half-plane and connectedness are adaptive while the testers here need only random samples from the image, so the techniques used for analyzing them are different. The tester for convexity in~\cite{BMR16socg} uses only random samples, but it is not based on dynamic programming.


\subparagraph{Open questions.}
In this paper we give tolerant testers for several important problems on images.
It is open whether these testers are optimal. No nontrivial lower bounds are known for these problems. (For any non-trivial property, an easy lower bound on the query complexity of a distance approximation algorithm is $\Omega(1/\eps^2)$. This follows from the fact that $\Omega(1/\eps^2)$ coin flips are needed to distinguish between a fair coin and a coin that lands heads with probability $1/2+\eps$.) Thus, our testers for half-plane and convexity are nearly optimal in terms of query complexity (up to a logorithmic factor in $1/\eps$). But it is open whether their running time can be improved.

\ifnum\full=1
\subparagraph{Organization.} We give formal definitions and notation in Section~\ref{sec:defintions_notation}. Algorithms for being a half-plane, convexity, and connectedness are given in Sections~\ref{sec:half-plane},~\ref{sec:convexity}, and~\ref{sec:connectedness}, respectively. The sections presenting algorithms for being a half-plane and convexity start by giving a distance approximation algorithm and conclude with the corollary about the corresponding PAC learner.
\else
\subparagraph{Organization.} We give formal definitions and notation in Section~\ref{sec:defintions_notation}, deferring some standard definitions to the {\color{black} full version}. Algorithms for being a half-plane, convexity, and connectedness are given in Sections~\ref{sec:half-plane},~\ref{sec:convexity}, and~\ref{sec:connectedness}, respectively.  We view our half-plane result as a good preparation for our distance approximation algorithm for convexity, the most technically difficult result in the paper. Corollaries about PAC learners as well as all omitted proofs and numerous figures can be found in the {\color{black} full version}.
\fi

\section{Definitions and Notation}\label{sec:defintions_notation}
We use $\integerset{n}$ to denote the set of integers $\{0,1,\ldots,n-1\}$ and $[n]$ to denote $\{1,2,\ldots,n\}$.
\ifnum\full=1
By $\log$ we mean the logarithm base 2, and by $\ln$, the logarithm base $e$.
\fi

\subparagraph{Image representation.}
We focus on black and white images. For simplicity, we only consider square images, but everything in this paper can be easily generalized to rectangular images.
We represent an image by an $n\times n$ binary matrix $M$ of pixel values, where 0 denotes white and 1 denotes black.
We index the matrix by $\integerset{n}^2$. The object is a subset of $\integerset{n}^2$
corresponding to black pixels; namely, $\{(i,j)\mid M[i,j]=1\}$.
\ifnum\full=1
The {\em left border of the image} is the set $\{(0,j)\mid j\in\integerset{n}\}$. The right, top and bottom borders are defined analogously. The image {\em border} is the set of pixels on all four borders.

For any region $R$, we use $A(R)$ to denote its area.
\fi

\ifnum\full=0
The {\em absolute distance}, $\Dis(M_1,M_2)$, between matrices $M_1$ and $M_2$ is the number of the entries on which they differ. The \emph{relative distance} between them is $\dis(M_{1},M_{2})=\Dis(M_{1},M_{2})/ n^{2}$.
A property $\mathcal{P}$ is a subset of binary matrices.
\else
\subparagraph{Distance to a property.} The {\em absolute distance}, $\Dis(M_1,M_2)$, between matrices $M_1$ and $M_2$ is the number of the entries on which they differ. The \emph{relative distance} between them is $\dis(M_{1},M_{2})=\Dis(M_{1},M_{2})/ n^{2}$.
A property $\mathcal{P}$ is a subset of binary matrices. The distance of an image represented by matrix $M$ to a property $\mathcal{P}$ is $\dis(M,\mathcal{P})=\min_{M'\in \mathcal{P}}$ $\dis(M,M')$. An image is {\em $\epsilon$-far} from the property if its distance to the property is at least $\epsilon$; otherwise, it is $\eps$-close to it.

\subparagraph{Computational Tasks.}
We consider several computational tasks: tolerant testing~\cite{PRR06}, additive approximation of the distance to the property, and proper (agnostic) PAC learning~\cite{Valiant84,KearnsSS94,Haussler92}.
\ifnum\full=1
Here we define them specifically for properties of images.
\fi
\begin{definition}[Tolerant tester]\label{def:tester}
 An {\em $(\eps_1,\eps_2)$-tolerant tester} for a property $\cP$ is a randomized algorithm that, given two parameters $\eps_1,\eps_2\in(0,1/2)$ such that $\eps_1<\eps_2$ and access to an $n\times n$ binary matrix $M$,
\ifnum\full=0
(1) accepts with probability at least 2/3 if $\dis(M,\cP)\leq\eps_1$;
(2) rejects with probability at least 2/3 if $\dis(M,\cP)\geq\eps_2$.
\else
\begin{enumerate}
\item accepts with probability at least 2/3 if $\dis(M,\cP)\leq\eps_1$;
\item rejects with probability at least 2/3 if $\dis(M,\cP)\geq\eps_2$.
\end{enumerate}
\fi

\end{definition}


\begin{definition}[Distance approximation algorithm]\label{def:distance-approx-alg}
An {\em $\mydelta$-additive distance approximation algorithm} for a property $\cP$ is a randomized algorithm that, given an error parameter $\mydelta\in(0,1/4)$ and access to an $n\times n$ binary matrix $M,$ outputs a value $\dout\in[0,1/2]$ that with probability at least 2/3 satisfies $|\dout -\dis(M,\cP)|\leq \mydelta$.
\end{definition}

As observed in \cite{PRR06}, we can obtain an $(\eps_1,\eps_2)$-tolerant tester for any property $\cP$ by running a distance approximation algorithm for $\cP$ with $\mydelta=(\eps_2-\eps_1)/2$. Thus, all our distance approximation algorithms directly imply tolerant testers.

\begin{definition}[Proper agnostic PAC learner]

A proper agnostic PAC learning algorithm for class $\cal P$ that works under the uniform distribution is given a parameter $\eps\in(0,1/2)$ and access to an image $M$. It can draw independent uniformly random samples $(i,j)$  and obtain $(i,j)$ and $M[i,j]$. With probability at least 2/3, it must output an image $M'\in {\cal P}$ such that $\dis(M,M')\leq\dis(M,{\cal P})+\eps$.
\end{definition}
\fi
\subparagraph{Access to the input.}
A {\em query-based} algorithm accesses its $n\times n$ input matrix $M$ by specifying a query pixel $(i,j)$ and obtaining $M[i,j]$.
\ifnum\full=1
The query complexity of the algorithm is the number of pixels it queries.
A {\em query-based} algorithm is {\em adaptive} if its queries depend on answers to previous queries and {\em nonadaptive} otherwise.
\fi
A {\em uniform} algorithm accesses its $n\times n$ input matrix by drawing independent samples $(i,j)$ from the uniform distribution over the domain (i.e., $\integerset{n}^2$) and obtaining $M[i,j]$.
A {\em block-uniform algorithm} accesses its $n\times n$ input matrix by specifying a block length $\side\in[n]$. For a block length $\side$ of its choice, the algorithm draws $x,y\in[\lceil n/\side\rceil]$ uniformly at random and obtains set $\{(i,j)\mid \lfloor i/\side\rfloor=x\text{ and } \lfloor j/\side\rfloor=y\}$ and $M[i,j]$ for all $(i,j)$ in this set.
The sample complexity of a {\em uniform} or a {\em block-uniform} algorithm is the number of pixels of the image it examines.

\begin{remark}\label{remark:bernoulli}
Uniform algorithms have access to independent (labeled) samples from the uniform distribution over the domain.
\ifnum\full=1
Sometimes it is more convenient to design
{\em Bernoulli algorithms} that
\else
{\em Bernoulli algorithms}
\fi
only have access to (labeled) Bernoulli samples from the 
image: namely, each pixel appears in the sample with probability $s/n^2$, where $s$ is the sample parameter that controls the expected sample complexity.
By standard arguments, a Bernoulli algorithm with the sample parameter $s$ can be used to obtain a uniform algorithm that takes $O(s)$ samples and has the same guarantees as the original algorithm (and vice versa).
\end{remark}

\section{Distance Approximation to the Nearest Half-Plane Image}\label{sec:half-plane}
{\color{black} 
{\color{black} An image is called a \emph{half-plane image}} if there exist an angle $\varphi\in[0,2\pi)$ and a real number $c$ such that pixel $(x,y)$ is black in the image iff $x\cos\varphi +y\sin\varphi\geq c$. The line $x\cos\varphi +y\sin\varphi= c$, denoted $\sepline{\varphi}{c}$, is {\em a separating line} of the half-plane image, i.e., it
separates black and white pixels of the image. We call $\varphi$ the {\em direction} of the half-plane image (and $\sepline{\varphi}{c}$). Note that $\varphi$ is the oriented angle between the $x$-axis and a line perpendicular to $\sepline{\varphi}{c}$. For all $\varphi\in[0,2\pi)$ and $c\in\mathbb R$, the half-plane image with a separating line $\sepline{\varphi}{c}$ is denoted $\hpi{\varphi}{c}$ and the closed half-plane whose every point $(x,y)$ satisfies the inequality $x\cos\varphi +y\sin\varphi \geq c$ is denoted $\hp{\varphi}{c}$. We can think of a half-plane image as a discretized half-plane.
 
}
\begin{theorem}\label{thm:half-plane-dist-appr}
For $\eps\in(\frac{90}n,\frac 1 4)$, there is a uniform $\mydelta$-additive distance approximation algorithm for the half-plane property with sample complexity  $O(\frac{1}{\mydelta^2}\log\frac{1}{\mydelta})$ and time complexity $O(\frac{1}{\mydelta^{3}}\log \frac{1}{\mydelta}).$
\end{theorem}

\begin{proof}
{\color{black}At a high level, our algorithm for approximating the distance to being a half-plane (Algorithm~\ref{alg:half-plane}) constructs a small set $\mathcal M_\mydelta$ of reference half-plane images. It samples pixels uniformly at random and outputs the empirical distance to the closest reference half-plane image. The core property of $\mathcal M_\mydelta$ is that the smallest empirical distance to a half-plane image in $\mathcal M_\mydelta$ can be computed quickly.

\begin{definition}[Reference directions and half-planes]\label{def:reference-half-planes}
Given $\mydelta\in(0,\frac{1}{4})$, let $a=\mydelta n/\sqrt{2}$. Let $D_\mydelta$ be the set of directions of the form $i\mydelta$ for  $i\in\integerset{\lceil 2\pi/\mydelta\rceil},$ called {\em reference directions}. The set of {\em reference half-plane images}, denoted $\mathcal M_\mydelta$, consists of every half-plane image for which $\sepline{\varphi}{c}$ is a separating line, where $\varphi\in D_\mydelta$ and $c$ is an integer multiple of $a$. 
\end{definition}
In other words, for every reference direction, we space separating lines of reference half-plane images distance $a$ apart. By definition, there are at most $\sqrt 2 n/a = 2/\mydelta$ reference half-plane images for each direction in $D_\mydelta$ and, consequently, $|\mathcal M_\mydelta|\leq 2 \pi/\mydelta \cdot(2/\mydelta) < 13/\mydelta^2$.

\begin{algorithm}
\caption{Distance approximation to being a half-plane.}
\label{alg:half-plane}
\SetKwInOut{Input}{input}\SetKwInOut{Output}{output}
\Input{parameters $n\in\mathbb{N}$, $\mydelta\in(90/n,1/4)$; Bernoulli access to an $n\times n$ binary matrix $M$.}
\DontPrintSemicolon
\BlankLine
\nl Sample a set $S$ of $s=\frac{6}{\mydelta^{2}}\ln\frac{7}{\mydelta}$ pixels uniformly at random with replacement.\;
\nl Let $D_\mydelta, \mathcal M_\mydelta$ be the sets of reference directions and half-planes, respectively (see Definition~\ref{def:reference-half-planes}) and $a=\eps n/\sqrt{2}$.\;
\tcp{Compute $\displaystyle\dout=\min_{M'\in \mathcal M_\mydelta}\dout(M')$,
where $\dout(M')=\frac 1 s \cdot |\{p\in S : M[p]\neq M'[p]\}|$:}
\nl\label{step:half-planes-foreach-dir} \ForEach {$\varphi\in D_\mydelta$}\do{\tcp{Lines with direction $\varphi$ partition the image. Bucket sort samples by position in the partition:}
\nl\quad     Assign each sample $(x,y)\in S$ to bucket $j=\lfloor(x\cos\varphi+y\sin\varphi)/a\rfloor$.\;
\nl\quad     For each bucket $j$, compute $w_j$ and $b_j$, the number of white and black pixels it has.\;
\nl\label{step:half-planes-compute-estimate}\quad     For each $j$, where $\hpi{\varphi}{ja}\in \mathcal M_\mydelta$, compute
$\dout(\hpi{\varphi}{ja})=
\frac{1}{s}\sum_{k < j}b_k+\
\frac{1}{s}\sum_{k\ge j}w_k$.\;
}
\nl Output $\dout$, the minimum of the values computed in Step~\ref{step:half-planes-compute-estimate}.\;
\end{algorithm}

\begin{lemma}\label{lem:properties-of-reference-HP}
 For every half-plane image $M$, there is $M'\in \mathcal M_\mydelta$ such that $\dis(M,M')\leq\mydelta/1.8$.
\end{lemma}

\begin{proof}
We mentioned that a half-plane image can be viewed as a discretized half-plane. Next we define a set of half-planes that we use in the proof of the lemma.
\begin{definition}\label{def:hp}
The set of reference half-planes, denoted $\mathcal H_{\mydelta}$, consists of every half-plane $\hp{\varphi}{c}$, where $\varphi\in D_\mydelta$ and $c$ is an integer multiple of $a$. 
\end{definition}
\begin{claim}\label{cl:properties-of-reference-HP}
For every half-plane $H$, there is a half-plane $H'\in\mathcal{H_{\eps}}$ such that the area of the symmetric difference of $H$ and $H'$ is at most $\mydelta n^2/2$.
\end{claim}
\begin{proof}
Consider a half-plane $\hp{\varphi}{c}$. Let $\varphi'$ be a reference direction closest to $\varphi$. Then $|\varphi-\varphi'|\le\mydelta/2$. 
We consider two cases. 
\ifnum\full=1
See Figures~\ref{fig:nearby-ref-halfplane} and~\ref{fig:nearby-ref-halfplane2}.

\begin{figure}[ht]
\begin{minipage}[b]{0.45\linewidth}
\centering
\includegraphics[width=\linewidth]{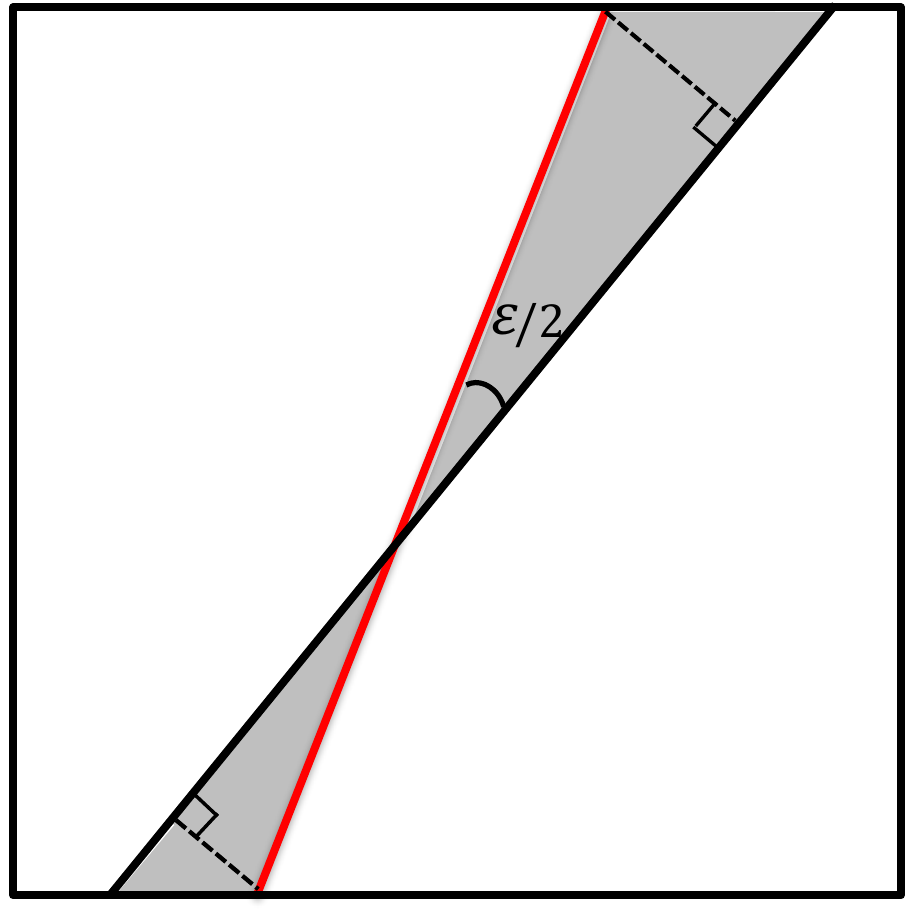}
\caption{Proof of Lemma~\ref{lem:properties-of-reference-HP}: triangular regions.}
\label{fig:nearby-ref-halfplane}
\end{minipage}
\hspace{0.1\linewidth}
\begin{minipage}[b]{0.45\linewidth}
\centering
\includegraphics[width=\linewidth]{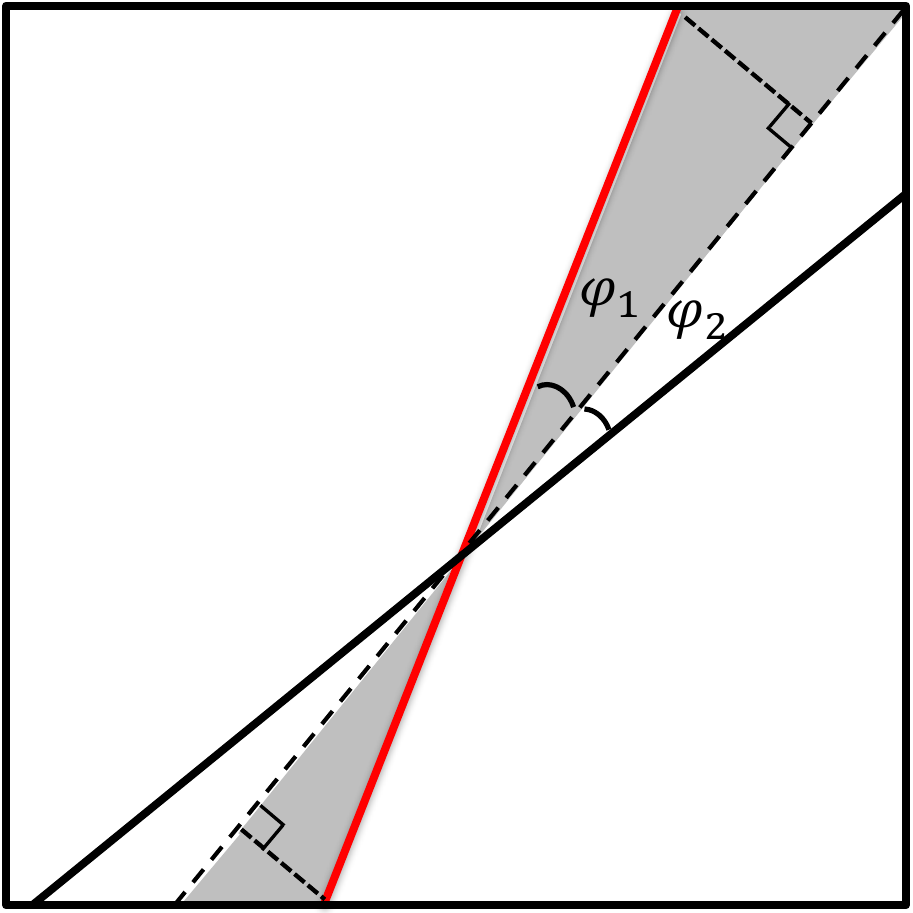}
\caption{Proof of Lemma~\ref{lem:properties-of-reference-HP}: triangular and quadrilateral regions.}
\label{fig:nearby-ref-halfplane2}
\end{minipage}
\end{figure}

\else
See the figures in the {\color{black} full version}.
\fi

{\bf Case 1:} Suppose that there is a reference half-plane $\hp{\varphi'}{c'}$ such that the lines $\sepline{\varphi}{c}$ and $\sepline{\varphi'}{c'}$ intersect inside inside $[0,n-1]^2$. Note that the length of every line segment inside $[0,n-1]^2$ is at most $\sqrt{2}n$. The symmetric difference of $\hp{\varphi}{c}$ and $\hp{\varphi'}{c'}$ inside $[0,n-1]^2$ consists of two regions
formed by lines $\sepline{\varphi}{c}$ and $\sepline{\varphi'}{c'}$. Each of these regions is either a triangle or (if it contains a corner of the image) a quadrilateral. First, suppose both regions are triangles. The sum of lengths of their bases, that lie on the same line, is at most $\sqrt{2}n$, whereas the sum of their heights is at most $\sin (\mydelta/2)\times\sqrt{2}n\leq \mydelta n/\sqrt{2}$. Hence, the sum of their areas is at most $\mydelta n^2/2$.

If exactly one of the regions is a quadrilateral, we add a line through the corner of the image contained in the quadrilateral and the intersection point of  $\sepline{\varphi}{c}$ and $\sepline{\varphi'}{c'}$. It partitions the symmetric difference of $\hp{\varphi}{c}$ and $\hp{\varphi'}{c'}$ into two pairs of triangular regions. Let $\varphi_1$ (respectively, $\varphi_2$) be the angle between the new line and $\sepline{\varphi}{c}$ (respectively, $\sepline{\varphi'}{c'}$). Then $\varphi_1+\varphi_2\leq\mydelta/2$. Applying the same reasoning as before to each pair of regions, we get that the sum of their areas is at most $\varphi_1 n^2+\varphi_2 n^2\leq \mydelta n^2/2$.
If both regions are quadrilaterals, we add a line as before for each of them and apply the same reasoning as before to the three resulting pairs of regions. Again, the area of the symmetric difference of $\hp{\varphi}{c}$ and $\hp{\varphi'}{c'}$ is at most $\mydelta n^2/2$. Thus, $\hp{\varphi'}{c'}$ is the required $M'$.

{\bf Case 2:} There exist reference half-planes $\hp{\varphi'}{c'}$ and $\hp{\varphi'}{c'+a}$ such that the line 
$\sepline{\varphi}{c}$ is between $L=\sepline{\varphi'}{c'}$ and $L'=\sepline{\varphi'}{c'+a}$. The region between $L$ and $L'$ inside the image has
length at most $\sqrt{2}n$ and width $a$. Thus, its area is at
most $\mydelta n^2$. Partition it into two regions: between $L$ and
$\sepline{\varphi}{c}$ and between $L'$ and $\sepline{\varphi}{c}$. One of the two regions has area at most
$\mydelta n^2/2$. Thus, $\hp{\varphi'}{c'}$ or $\hp{\varphi'}{c'+a}$ is the required $M'$.
\end{proof}

To complete the proof of the lemma we use the following theorem that relates the area of a lattice polygon and the number of integer points that the polygon covers. (A lattice polygon is a polygon whose vertices have integer coordinates.)
\begin{theorem}[Pick's theorem~\cite{pick}]\label{thm:pick}
For a simple lattice polygon $G$, let $\alpha$ denote the number of lattice points in the interior of $G$ and $\beta$ denote the number of lattice points on the boundary of $G$. Then $A(G)=\alpha+\beta/2-1$.
\end{theorem}

\begin{definition}
For a polygon $G$, let $Perim(G)$ denote the perimeter of $G$ and $Pix(G)$ denote the number of pixels in $G$, i.e., pixels in the interior of $G$ and on its boundary.
\end{definition}
\begin{proposition}\label{prop:area-pixel}
Let $G$ be a convex polygon. Then $Pix(G)\leq A(G)+Perim(G)/2+1$.
\end{proposition}
\begin{proof}
If all pixels in $G$ are collinear then $Pix(G)\leq Perim(G)/2+1\leq A(G)+Perim(G)/2+1$. This follows from the fact that the length of a line segment inside a polygon is at most half of the perimeter of the polygon and that the number of integer points on the line segment is at most the length of the line segment plus one. If not all pixels in $G$ are collinear then consider the convex hull of all pixels in $G$. Let $\alpha$ and $\beta$ denote the number of pixels in the interior and on the boundary of that convex hull, respectively. (Note that the convex hull is a lattice polygon). By Theorem~\ref{thm:pick}, we obtain that $\alpha+\beta/2-1\leq A(G)$ and $Pix(G)=\alpha+\beta\leq A(G)+\beta/2+1\leq A(G)+Perim(G)/2+1$.
\end{proof}

For some $\varphi$ and $c$, half-plane image $M=\hpi{\varphi}{c}$. Consider the half-plane $\hp{\varphi}{c}$. By Claim~\ref{cl:properties-of-reference-HP}, there is a half-plane $\hp{\varphi'}{c'}$ such that the area of the symmetric difference of $\hp{\varphi}{c}$ and $\hp{\varphi'}{c'}$ is at most $\mydelta n^2/2$, where $\varphi'\in D_{\mydelta}$ and $c$ is  a multiple of $a$. 

Recall that there are 4 cases for the symmetric difference of $H$ and $H'$. More precisely, it consists of: 1) two triangles, 2) a triangle and a quadrilateral, 3) a quadrilateral, or 4) two quadrilaterals. We consider the last case (this is the hardest case and the three other cases are handled similarly). Let the symmetric difference of $H$ and $H'$ consist of two quadrilaterals $Q_1$ and $Q_2$. (For reference, see Figure~\ref{fig:nearby-ref-halfplane2} where a triangle and a quadrilateral are shown.) Every line segment in the image has length at most $\sqrt{2}n$. Thus, $Perim(Q_1)+Perim(Q_2)\leq 6\sqrt{2}n$. By Proposition~\ref{prop:area-pixel}, we obtain that $Pix(Q_1)+Pix(Q_2)\leq A(Q_1)+A(Q_2)+(Perim(Q_1)+Perim(Q_2))/2+2\leq \eps n^2/2+3\sqrt{2}n+2\leq\eps n^2/1.8$ (recall that $\eps\in(90/n,1/4)$). This completes the proof.
\end{proof}
}
\subparagraph{Analysis of Algorithm~\ref{alg:half-plane}.}
Let $d_M$ be the distance of $M$ to being a half-plane. Then there exists a half-plane matrix $M^{*}$ such that $\dis(M,M^{*})=d_M$.
By a uniform convergence bound (see, e.g., \cite{Avrim-lecture-notes}), since $s\geq (2.6/\mydelta^2)(\ln |M_\mydelta| + \ln 6)$ for all $\mydelta\in(0,1/4)$, we get that with probability at least 2/3,
$|\dis(M,M')-\dout(M')|\leq \mydelta/2.25$ for all $M'\in M_\mydelta$. Suppose this event happened. Then $\dout\geq d_M-\mydelta/2.25$ because $\dis(M,M')\geq d_M$ for all half-planes $M'$. Moreover, by Lemma~\ref{lem:properties-of-reference-HP}, there is a matrix $\hat{M}\in H_{\mydelta}$
such that $\dis(M,\hat{M})\leq\dis(M,M^*)+\dis(M^*,\hat{M})
\leq d_M+\mydelta/1.8$. For this matrix, $\dout(\hat{M})\leq\dis(M,\hat{M})+\mydelta/2.25\leq d_M+\mydelta.$ Thus,
$d_M-\mydelta/2.25\leq\dout\leq d_M+\mydelta.$ That is, $|d_M-\dout|\leq \mydelta$ with probability 2/3, as required.

\subparagraph{Sample and time complexity.} The number of samples, $s$, is $O(1/\mydelta^2 \log 1/\mydelta)$. To analyze the running time, recall that $|D_\mydelta|=O(1/\mydelta)$. For each direction in $D_\mydelta$, we perform a bucket sort of all samples in expected $O(s)$ time. The remaining steps in the {\bf foreach} loop of Step~\ref{step:half-planes-foreach-dir} can also be implemented to run in $O(s)$ time. The expected running time of Algorithm~\ref{alg:half-plane} is thus $O(1/\mydelta \cdot s)=O(1/\mydelta^3\log 1/\mydelta)$.
Remark~\ref{remark:bernoulli} implies a tester with the same worst case running time.
\end{proof}


\ifnum\full=1
\begin{corollary}\label{cor:half-plane-agnostic-learner}
The class of half-plane images is properly agnostically PAC-learnable with sample complexity  $O(\frac{1}{\mydelta^2}\log\frac{1}{\mydelta})$ and time complexity $O(\frac{1}{\mydelta^{3}}\log \frac{1}{\mydelta})$ under the uniform distribution.
\end{corollary}
\begin{proof}
We can modify Algorithm~\ref{alg:half-plane} to output, along with $\dout=\min_{M'\in M_\mydelta}\dout(M')$, a reference half-plane $\hat{M}$ that minimizes it. By the analysis of Algorithm~\ref{alg:half-plane}, with probability at least 2/3, the output $\hat{M}$ satisfies $\dis(M,\hat{M})\leq d_M+\mydelta.$
\end{proof}
\fi

\section{Distance Approximation to the Nearest Convex Image}\label{sec:convexity}
An image is \emph{convex} if the convex hull of all black pixels contains only black pixels.
%
\begin{theorem}\label{thm:convexity_dist_appr}
For $\eps\in(n^{-1/6},1/4)$, there is a uniform $\mydelta$-additive distance approximation algorithm for convexity with sample complexity $O(\frac{1}{\mydelta^{2}}\log \frac 1 \mydelta)$ and running time $O(\frac 1 {\mydelta^8})$.
\end{theorem}
\begin{proof}
The starting point for our algorithm for approximating the distance to convexity (Algorithm~\ref{alg:convexity-dist-approximation}) is similar to that of Algorithm~\ref{alg:half-plane} that approximates the distance to a nearest half-plane. We define a small set $P_\mydelta$ of reference polygons. Algorithm~\ref{alg:convexity-dist-approximation} 
implicitly learns a nearby reference polygon and outputs the empirical distance from the image to that polygon. The key features of $P_\mydelta$ is that (1) every convex image has a nearby polygon in $P_\mydelta$, and (2) one can use dynamic programming (DP) to quickly compute the smallest empirical distance to a polygon in $P_\mydelta$.

We start by defining reference directions, lines, points, and line-point pairs that are later used to specify our DP instances.  
Reference directions are almost the same as in Definition~\ref{def:reference-half-planes}.

\begin{definition}[Reference lines, line-point pairs]\label{def:reference-lines}
Fix $\bigdelta=\mydelta/\con$.  The set of {\em reference directions} is $D_\mydelta=\{\pi/2\}\cup\{i\bigdelta:~i\in [0,\lceil 2\pi/\bigdelta\rceil)\}.$
For every $\varphi\in D_\mydelta$, define the set of {\em reference lines} $L_\varphi = \{\ell : \ell \text{ passes through the image and satisfies the equation }x\cos\varphi +y\sin\varphi = c, $\ where $c$ is an integer multiple of  $\bigdelta n$$\}$.
For each reference line, the set of {\em reference points on $\ell$} contains points w.r.t.\ $\ell$, which are inside $[0,n-1]^2$, spaced exactly $\bigdelta n$  apart (it does not matter how the initial point is picked). A {\em line-point pair} is a pair $(\ell,b),$ where $\ell$ is a 
reference line and $b$ is a reference point w.r.t.\ $\ell$. (Note that there could be reference points on $\ell$ that were defined w.r.t.\ some other reference line. This is why we say ``a reference point w.r.t.\ $\ell$'', and not ``a reference point on $\ell$''.)
\end{definition}

Roughly speaking, a reference polygon is a polygon whose vertices are defined by line-point pairs. There are additional restrictions that stem from the fact that we need to be able to efficiently find a nearby reference polygon for an input image. The actual definition specifies which actions we can take while constructing a reference polygon. Reference polygons are built starting from reference boxes, which are defined next.

\ifnum\full=1
\begin{figure}[ht]
\begin{minipage}[b]{0.45\linewidth}
\centering
\includegraphics[width=\linewidth]{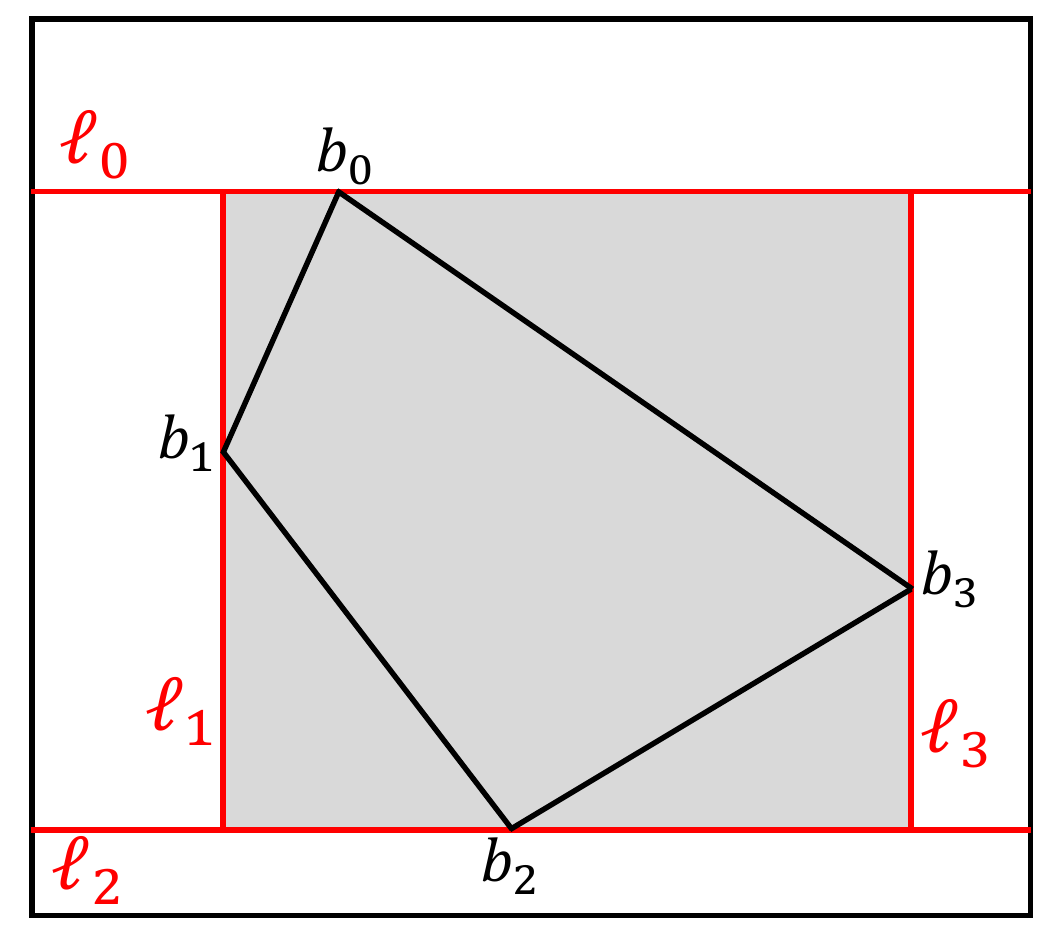}
\caption{ A reference box.}
\label{fig:bounding-box}
\end{minipage}
\hspace{0.1\linewidth}
\begin{minipage}[b]{0.45\linewidth}
\centering
\includegraphics[width=\linewidth]{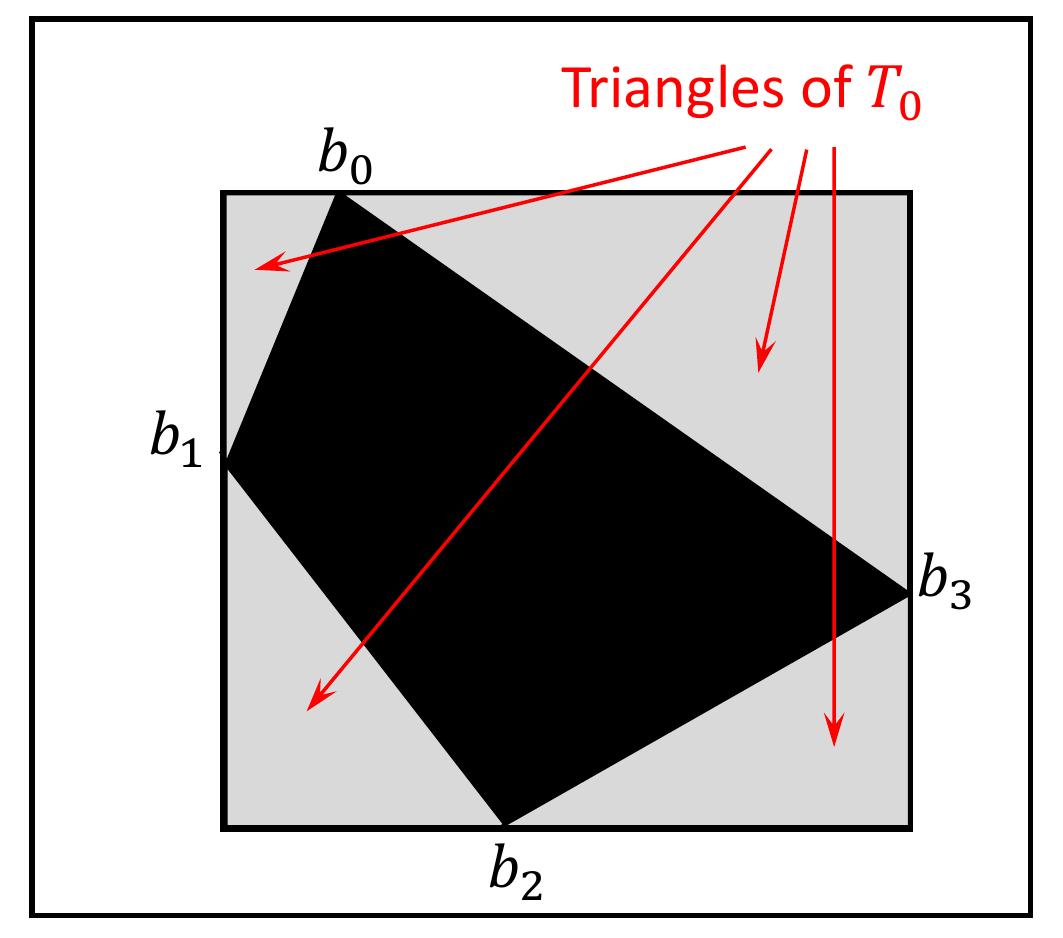}
\caption{ Triangles of the set $T_0$.}
\label{fig:triangles-T0}
\end{minipage}
\end{figure}
\fi

\begin{definition}[Reference box]\label{def:reference-box}
A {\em reference box} is a set of four line-point pairs $(\ell_i,b_i)$ for $i=0,1,2,3$, where $\ell_0,\ell_2$ are distinct horizontal lines, such that $\ell_0$ is above $\ell_2$, and $(\ell_1,\ell_3)$ are distinct vertical lines, such that $\ell_1$ is to the left of $\ell_3$. The reference box defines a vertex set $B_0=\{b_0,b_1,b_2,b_3\}$ and a triangle set $\Tstart,$ formed by removing the quadrilateral $b_0b_1b_2b_3$ from the rectangle delineated by the lines $\ell_0,\ell_1,\ell_2,\ell_3$\ifnum\full=1.
See Figures~\ref{fig:bounding-box}-~\ref{fig:triangles-T0}.
\else
.
\fi
\end{definition}
\ifnum\full=1
Note that line-point pairs do not depend on the input.
\else

\fi
Intuitively, by picking a reference box, we decide to keep the area inside the quadrilateral $b_0b_1b_2b_3$ black, the area outside the rectangle formed by $\ell_0,\ell_1,\ell_2,\ell_3$ white, and the triangles in $\Tstart$ gray, i.e., undecided for now.

\begin{definition}
For points $x,y$, let $\ell(x,y)$ denote the line that passes through $x$ and $y$. Let $xy$ denote the line segment between $x$ and $y$\ifnum\full=0 . \else\ and $|xy|$ denote the length of $xy$.\fi
\end{definition}

Reference polygons are defined next.  Intuitively, to obtain a reference polygon, we keep subdividing ``gray'' triangles in $\Tstart$ into smaller triangles and deciding to color the smaller triangles black or white or keep them gray (i.e., undecided for now). We also allow ``cutting off'' a quadrilateral that is adjacent to black and coloring it black (a.k.a.\ ``the base change operation'').
\ifnum\full=1
The main recoloring operation from Definition~\ref{def:reference-polygons} is illustrated in Figure~\ref{fig:tolerant-reducing-area}.
\fi
Even though the definition of reference polygons is somewhat technical, the readers can check their understanding of this concept by following Algorithm~\ref{alg:convexity-dist-approximation}, as it chooses the best reference polygon to approximate the input image.

\begin{definition}[Reference polygon]\label{def:reference-polygons}
A {\em reference polygon} is an image of a polygon $\hull(B),$ where the set
$B$ can be obtained from a reference box with a vertex set $B_0$ and a triangle set $\Tstart$ by the following recursive process.
Initially, $\Tend=\emptyset$ and $B=B_0$. 
While $\Tstart\neq\emptyset,$ move a triangle $T$ from $\Tstart$ to $\Tend$ and
perform the following steps:

\begin{enumerate}
\item\label{item:ref-poly-definition} {\sf (Base Change).} Let $T=\bigtriangleup b'b''v,$ where
 $b',b''\in B.$ Select reference point $b'_0$ on $b'v$ w.r.t.\ line $\ell(b',v)$, and reference point $b''_0$ on $b''v$ w.r.t.\ line $\ell(b'',v)$. Add $b'_0,b''_0$ to $B$. (This corresponds to coloring the quadrilateral $b'b'_0b''_0b''$ black.) Let $h$ be the height of $\bigtriangleup b'_0b''_0v$ w.r.t.\ the base $b'_0b''_0.$


\item\label{item:ref-poly-definition-tall} {\sf (Subdivision Step)} If $h>6\bigdelta n$, choose whether to proceed with this step or go to Step~\ref{item:ref-poly-definition-short} (both choices correspond to a legal reference polygon); otherwise, go to Step~\ref{item:ref-poly-definition-short}. Let $\varphi$ be the angle between $\ell(b'_0,b''_0)$ and the $x$-axis, and $\hat\varphi\in D_\mydelta$ be such that $|\hat\varphi-\varphi|\leq {\bigdelta/2}$. Select a reference line-point pair $(\ell,b),$ where the line $\ell\in L_{\hat{\varphi}}$ crosses $b'_0v$ and $b''_0v$, whereas $b$ is in the triangle $\bigtriangleup b'_0b''_0v$. Let $v'$ (resp., $v''$) be the point of intersection of $\ell$ and $b'_0v$ (resp., $\ell$ and $b''_0v$). Let $T'=\bigtriangleup b_0'bv'$, $T''=\bigtriangleup b_0''bv''$%
\ifnum\full=1
, as shown on Figure~\ref{fig:tolerant-reducing-area}%
\fi
. Add $b$ to $B$ and triangles $T',T''$ to $\Tstart$. (This represents coloring $\bigtriangleup b'_0b''_0b$ black and keeping $T'$ and $T''$ gray.)

\item\label{item:ref-poly-definition-short} {\sf (End of Processing)} Do nothing. (This represents coloring $\bigtriangleup b'_0b''_0v$ white).
\end{enumerate}
\end{definition}
\ifnum\full=1
\begin{figure}[ht]
\begin{minipage}[b]{0.55\linewidth}
\centering
\includegraphics[width=\linewidth]{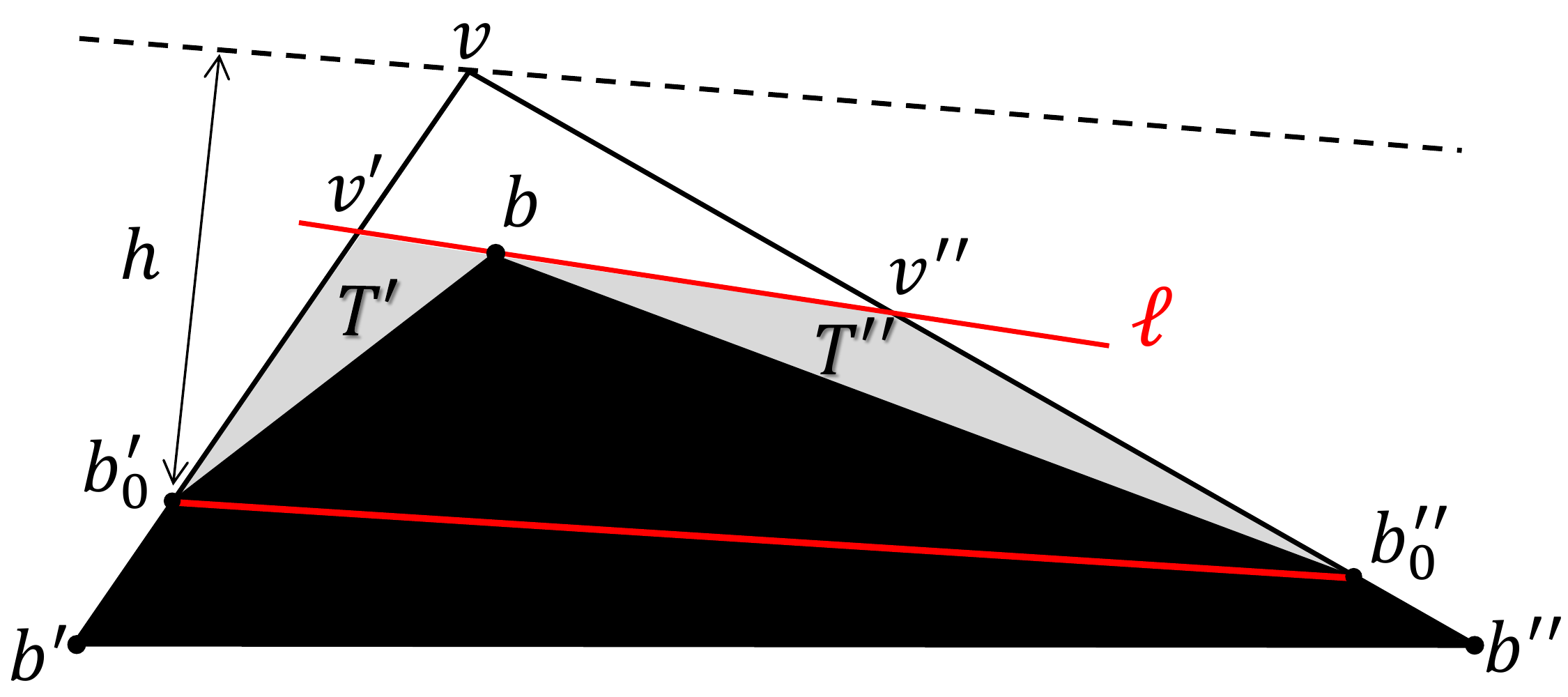}
\caption{ An illustration to Definition~\ref{def:reference-polygons}: Triangle $\bigtriangleup b'b''v$.}
\label{fig:tolerant-reducing-area}
\end{minipage}
\hspace{0.01\linewidth}
\begin{minipage}[b]{0.48\linewidth}
\centering
\includegraphics[width=\linewidth]{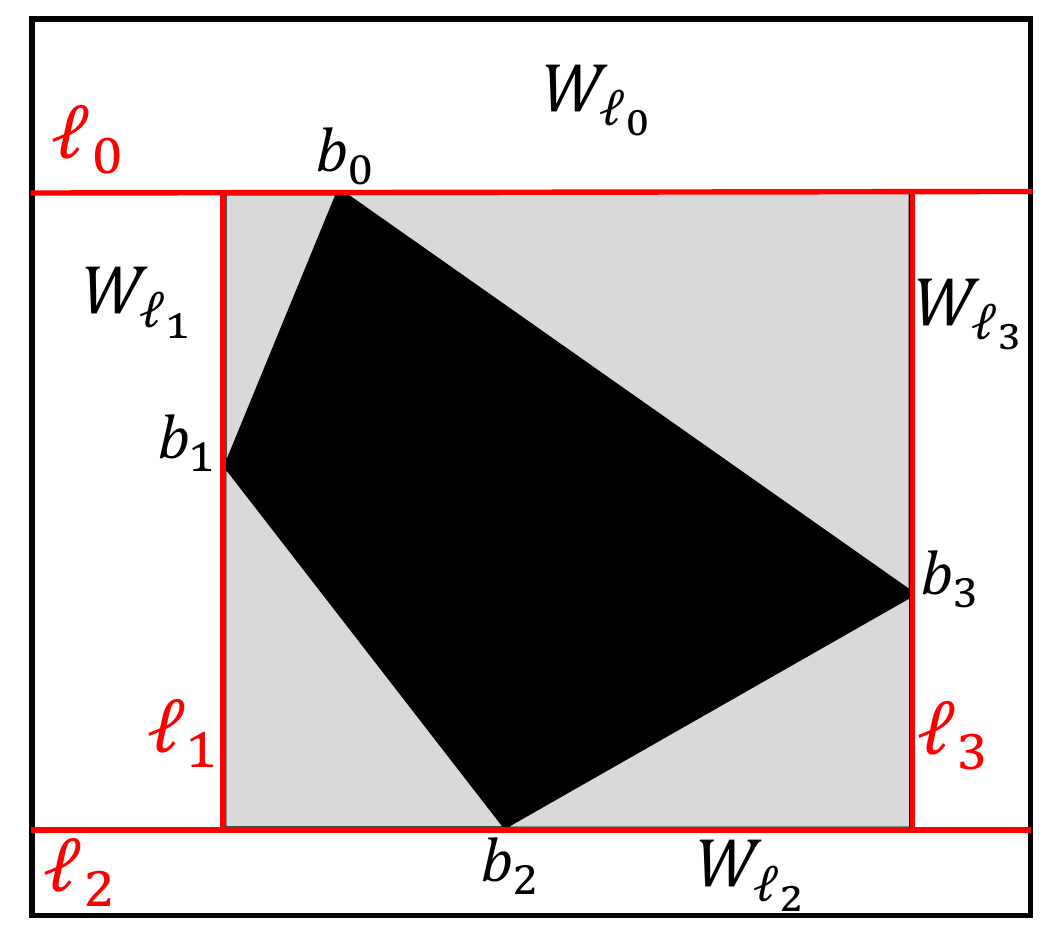}
\caption{ Regions $W_{\ell_0}$, $W_{\ell_1}$, $W_{\ell_2}$, and $W_{\ell_3}$.}
\label{fig:w-regions}
\end{minipage}
\end{figure}
\fi

By Remark~\ref{remark:bernoulli}, to prove Theorem~\ref{thm:convexity_dist_appr}, it suffices to design a Bernoulli tester that takes $s=O(\frac 1 {\mydelta^2}\log \frac 1\mydelta)$ samples in expectation and runs in time $O(\frac 1 {\mydelta^8})$. Our Bernoulli tester is Algorithm~\ref{alg:convexity-dist-approximation}.
In Algorithm~\ref{alg:convexity-dist-approximation}, we use the following notation for the (relative) empirical error with respect to an input image $M$, a set of sampled pixels $S,$ and the size parameter $s$. For an image $M'$, let
$\dout(M')=\frac 1 s \cdot |\{u\in S : M[u]\neq M'[u]\}|.$ For every region $R\subseteq\domain$, we let
$\dout_+(R)=\frac 1 s \cdot |\{u\in S\cap R : M[u]=0\}|,$ and
$\dout_-(R)=\frac 1 s \cdot |\{u\in S\cap R : M[u]=1\}|,$ i.e., the empirical error if we make $R$ black/white, respectively.

\begin{algorithm}
\caption{Bernoulli approximation algorithm for distance to convexity.}
\label{alg:convexity-dist-approximation}
\SetKwInOut{Input}{input}\SetKwInOut{Output}{output}
\Input{parameters $n\in\mathbb{N}$, $\mydelta\in(0,1/4)$; Bernoulli access to an $n\times n$ binary matrix $M$.}
\DontPrintSemicolon
\BlankLine
\nl\label{step:convexity-approx:sample}
Set $s=\Theta(\frac 1 {\mydelta^2}\log \frac 1\mydelta)$.
Include each image pixel in the sample $S$ w.p.\ $p=s/n^2$.\;
\tcp{Run the algorithm to find $\dout$, the smallest fraction of samples misclassified by a reference polygon in $P_\mydelta$. A dynamic programming implementation of the algorithm is given in \ifnum\full=0 Section~4.3 of the {\color{black} full version}.\else Section~\ref{sec:convexity-dist-approx-wrap-up}.\fi }
\nl
Let $W_{\ell_0}$ (resp., $W_{\ell_2}$) be the set of pixels of the image $M$ that lie either above $\ell_0$ or to the left of $b_0$ on $\ell_0$ (resp., either below $\ell_2$ or to the left of $b_2$ on $\ell_2$). Let $W_{\ell_1}$ (resp., $W_{\ell_3}$) be the set of pixels of $M-W_{\ell_0}-W_{\ell_2}$ to the left of $\ell_1$ (resp., to the right of $\ell_3$). \ifnum\full=1 (See Figure~\ref{fig:w-regions})\fi

\nl
Set $\dout=1$.
\;
\nl\label{step:polygons-foreach-top-bottom}
        \ForAll {line-point pairs $(\ell_0,b_0),(\ell_2,b_2)$, where $\ell_0,\ell_2$ are horizontal lines} \do{
\nl\quad\label{step:errle-init}     Set $\errle=1.$ \tcp*{\parbox[t]{4in}{The variable to compute the best error for the region to the left of $b_0b_2$, between $\ell_0$ and $\ell_2$.\;}}
\nl\quad\label{step:polygons-foreach-left}
\ForEach {line-point pair $(\ell_1,b_1)$, where $\ell_1$ is a vertical line} \do{
\nl\quad\quad Let $v_0$ (resp., $v_2$) be the point where $\ell_1$ intersects $\ell_0$ (resp., $\ell_1$ intersects $\ell_2$).\;
\nl\quad\quad\label{step:errle-compute}
               $\errle=\min (\errle,
                \dout_-(W_{\ell_1})+\dout_+(\bigtriangleup b_0b_1b_2)
                +\Best(\bigtriangleup b_0b_1v_0)+\Best(\bigtriangleup b_1b_2v_2))$\;
        }
\nl\quad        Similarly to Steps~\ref{step:errle-init}--\ref{step:errle-compute}, compute  $\errri$.\tcp*{\parbox[t]{10in}{The best error for the region to the right of $b_0b_2$, between $\ell_0$ and $\ell_2$.\;}}

\nl\quad    Compute $\dout=\min(\dout,
                    \dout_-(W_{\ell_0}\cup W_{\ell_2})
                    +\errle+\errri).$\;

    }
\nl         \Return $\dout$.
\end{algorithm}

\ifnum\full=1
Subroutine \Best, presented next,
\else
Subroutine \Best
\fi
chooses the option with the smallest empirical relative error among those given in Definition~\ref{def:reference-polygons}, items~\ref{item:ref-poly-definition}-\ref{item:ref-poly-definition-short}.
\ifnum\full=0
Its pseudocode is in the {\color{black} full version}.
\else
\begin{algorithm}
\caption{Subroutine \Best used in Algorithm~\ref{alg:convexity-dist-approximation}.}
\label{alg:subroutine-best}
\SetKwInOut{Input}{input}\SetKwInOut{Output}{output}
\Input{triangle $\bigtriangleup b'b''v$}
\DontPrintSemicolon
\BlankLine
\tcp{Use dynamic programming (see Section~\ref{sec:convexity-dist-approx-wrap-up} for implementation details).}
\nl Set $d^*=1.$

\nl\quad\label{step:forall-best-type2}\ForAll {reference points $b'_0$ and $b''_0$ on the sides $b'v$ and $b''v,$ respectively,}\do{

\nl\quad\quad Compute $d^*=\min(d^*, d_+(b'b''b''_0b'_0)+\BestFixed(b'_0b''_0v))$}

\nl\quad \Return $d^*$

\end{algorithm}

\begin{algorithm}\label{alg:subroutine-best-fixed-base}
\caption{Subroutine \BestFixed used in Algorithm~\ref{alg:subroutine-best}.}
\label{alg:subroutine-best-fixed-base}
\SetKwInOut{Input}{input}\SetKwInOut{Output}{output}
\Input{triangle $\bigtriangleup b'_0b''_0v$}
\DontPrintSemicolon
\BlankLine
\nl Set $d^*=d_-(\bigtriangleup b'_0b''_0v)$\;
\nl \If {the height of $\bigtriangleup b_1b_2v$ w.r.t.\ the base $b_1b_2$ is more than $6\bigdelta n$}{


\nl\quad\label{step:foreach-best-type1}\ForEach {line-point pair $(\ell,b)$, where $\ell\in L_{\hat\varphi}$ (see Definition~\ref{def:reference-polygons}, item 2), $b\in\bigtriangleup bb'_0b''_0$,
line $\ell$ intersects the side $b'_0v$ at some point $v'$ and the side $b''_0v$ at some point $v''$, resp.
} \do{

\nl\quad\quad Compute $d^*=\min(d^*, \dout_-(\bigtriangleup v'v''v)+\dout_+(\bigtriangleup b'_0b''_0b)
                +\Best(\bigtriangleup b'_0bv')+\Best(\bigtriangleup bb''_0v''))$
}}

\nl \Return $d^*$

\end{algorithm}
\fi

\ifnum\full=0
Our set of reference polygons has two critical features. First, for each convex image there is a nearby reference polygon. It turns out that the empirical error for a region is proportional to the square root of its area. The second key feature of our reference polygons is that, for each of them, the set of considered triangles, $\Tend$, has small $\sum_{T\in\Tend} \sqrt{A(T)},$ where $A(T)$ denotes the area of triangle $T$. The proofs of both features, as well as the analysis of the empirical error, are quite technical and appear in the {\color{black} full version}.\end{proof}

Here, we state and partially prove a lemma that puts together different parts of the analysis. It makes it clear why the empirical error of each region is proportional to the square root of its area which is, as explained in Footnote~\ref{fn:Pick}, a proxy for the number of pixels in it.

\else
Our set of reference polygons has two critical features. First, for each convex image there is a nearby reference polygon. This is proved in Section~\ref{sec:existence-of-nearby-ref-poly}. It turns out that the empirical error for a region is proportional to the square root of its area. The second key feature of our reference polygons is that, for each of them, the set of considered triangles, $\Tend$, has small $\sum_{T\in\Tend} \sqrt{A(T)}.$ The proof of this fact, as well as the analysis of the empirical error appears in Section~\ref{sec:convexity-dist-approx-error-analysis}. Finally, Section~\ref{sec:convexity-dist-approx-wrap-up} completes the analysis of the algorithm, gives details of its implementation and presents the corollary about agnostic PAC learning of convex objects.

\subsection{Existence of a nearby reference polygon}\label{sec:existence-of-nearby-ref-poly}
\begin{lemma}\label{lem:nearby-reference-polygon}
 For every convex image $M$, there exists $M'\in P_\mydelta$ such that $\dis(M,M')\leq\mydelta/6$.
\end{lemma}

\begin{proof}
Consider a convex image $M$. We will show how to construct a nearby reference polygon $M'$ using the recursive process in Definition~\ref{def:reference-polygons}. First, we obtain a reference box (see Definition~\ref{def:reference-box}) for $M$ as follows. Let $(\ell_0,b_0)$ be a line-point pair, where $b_0$ is black in $M$ and $\ell_0$ is the topmost horizontal line that contains such a reference point. Similarly, define $(\ell_2,b_2)$, replacing ``topmost'' with ``bottommost''. Analogously, define the two line-point pairs $(\ell_1,b_1)$, $(\ell_3,b_3)$ with vertical lines. The four line-point pairs $(\ell_i,b_i)$ for $i\in\integerset{4}$ define the reference box for $M$, as shown in Figure~\ref{fig:reference-box}.

\begin{figure}[ht]
\centering
\includegraphics[width=0.5\linewidth]{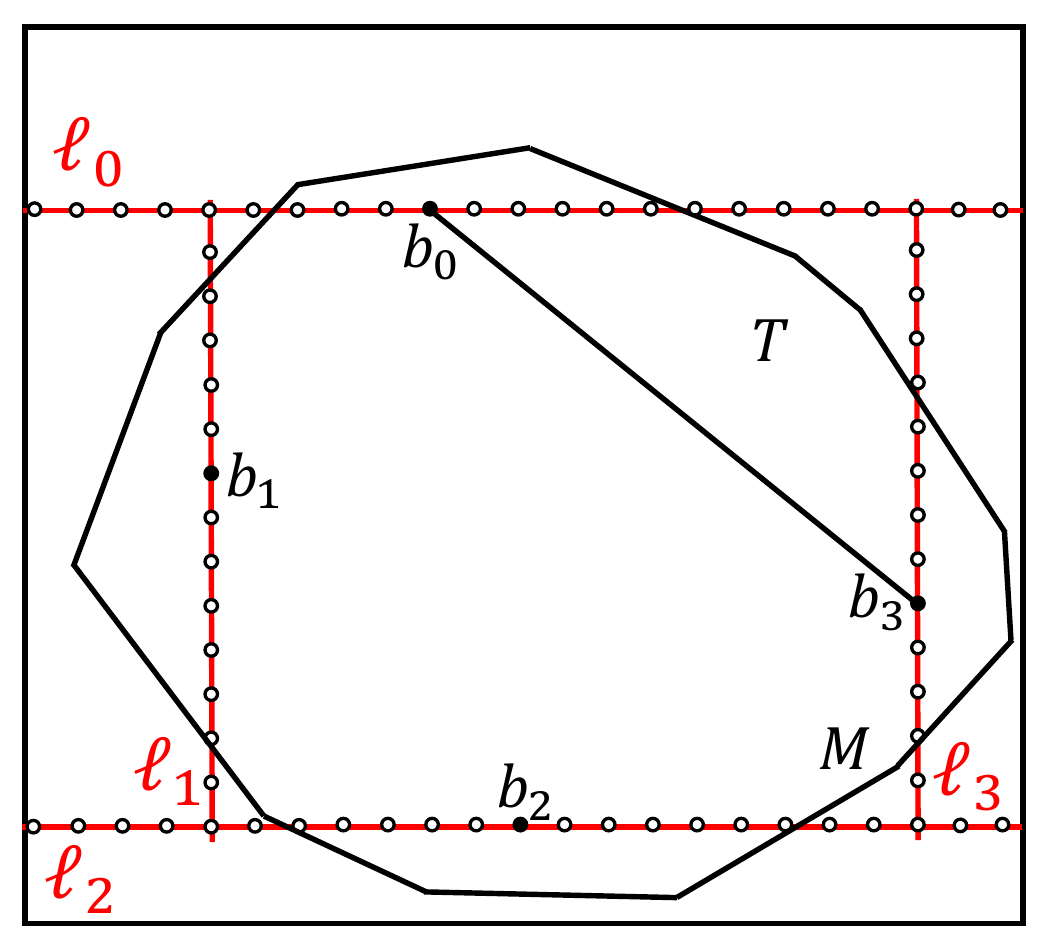}
\caption{ Reference box for a convex image $M$.}
\label{fig:reference-box}
\end{figure}

Next we construct the set $B$ from the reference box, as in Definition~\ref{def:reference-polygons}.  We also maintain two sets of line segments, $F_1$ and $F_2$, that are used in the analysis. Initially, $F_1=F_2=\emptyset$. The colors of the points in the description below are with respect to the convex image $M$. This is how we make the choices at each step of the recursive process in Definition~\ref{def:reference-polygons} to obtain our reference polygon:
\begin{enumerate}
\item {\sf (Base Change)} Choose $b'_0,b''_0$ to be the furthest from $b'b''$ black reference w.r.t.\ lines $\ell(b',v)$ and $\ell(b'',v)$, respectively. Recall that $h$ is the height of $\bigtriangleup b'_0b''_0v$ w.r.t.\ the base $b'_0b''_0$.

\item {\sf (Subdivision Step)} If $h>6\bigdelta n$, let $B_M$ denote the convex hull of all black pixels in $M$ and points $b'_1,b''_1$ be the intersection points of $B_M$ with $b'_0v$ and $b''_0v$, respectively. Choose a line-point pair $(\ell,b)$ such that $\ell\in L_{\hat\varphi}$ is the furthest from $b'_1b''_1$ line that intersects $b'v$ and $b''v$,  and $b$ is black.  Let $\ell$ intersect $b'v$ and $b''_1v$ at $v'$ and $v''$, respectively and let it intersect $B_M$ at $y'$ and $y'',$ as in Figure~\ref{fig:small-error-ref-poly-1}. Put the line segment $y'y''$ in $F_1$ and $\bigtriangleup v'v''v$ in $\Tcut$. If no line in $L_{\hat\varphi}$ contains a black reference point in $\bigtriangleup b'_1b''_1v$ or if $h\leq 6\bigdelta n,$ go to Step 3.

\item {\sf (End of Processing)} Put the line segment $b'_1b''_1$ in $F_2$ and $\bigtriangleup b'_0b''_0v$ in $\Tfin$. Triangle $\bigtriangleup b'_0b''_0v$ is not subdivided and is called a {\em final} triangle.
\end{enumerate}

Observe that $M$ and $M'$ differ only on three types of regions: outside of the reference box, inside the triangles in $\Tfin$, and inside the triangles in $\Tcut$. To show that $\Dis(M,M')\leq\frac{\mydelta n^2}{6},$ we prove in Claims~\ref{cl:error-in-strips}, \ref{cl:error-in-triangles}, and \ref{cl:small-area-above-line} that the number of disagreements in each of the three regions is small.
For any region $R\subseteq\domain$, let $\myerr{R}=|\{u\in R : M[u]\neq M'[u]\}|$.

Next claim
follows from the analysis of the convexity tester in \cite{Ras03}.
\begin{claim}
\label{cl:error-in-strips}
The number of black pixels in $M$
outside the reference box is at most $12\cdot\bigdelta n^2$.
\end{claim}
\mch{

\begin{claim}\label{cl:error-in-triangles}
Let $\bigtriangleup b'_0b''_0v$ be a final triangle and points $b'_1,b''_1$ be the points of intersection of $B_M$ with $b'_0v$ and $b''_0v$, respectively. Then $\myerr{\bigtriangleup b'_0b''_0v}\leq 4\cdot|b'_1b''_1|\bigdelta n+2$.
\end{claim}

\begin{figure}[ht]
\centering
\includegraphics[width=0.7\linewidth]{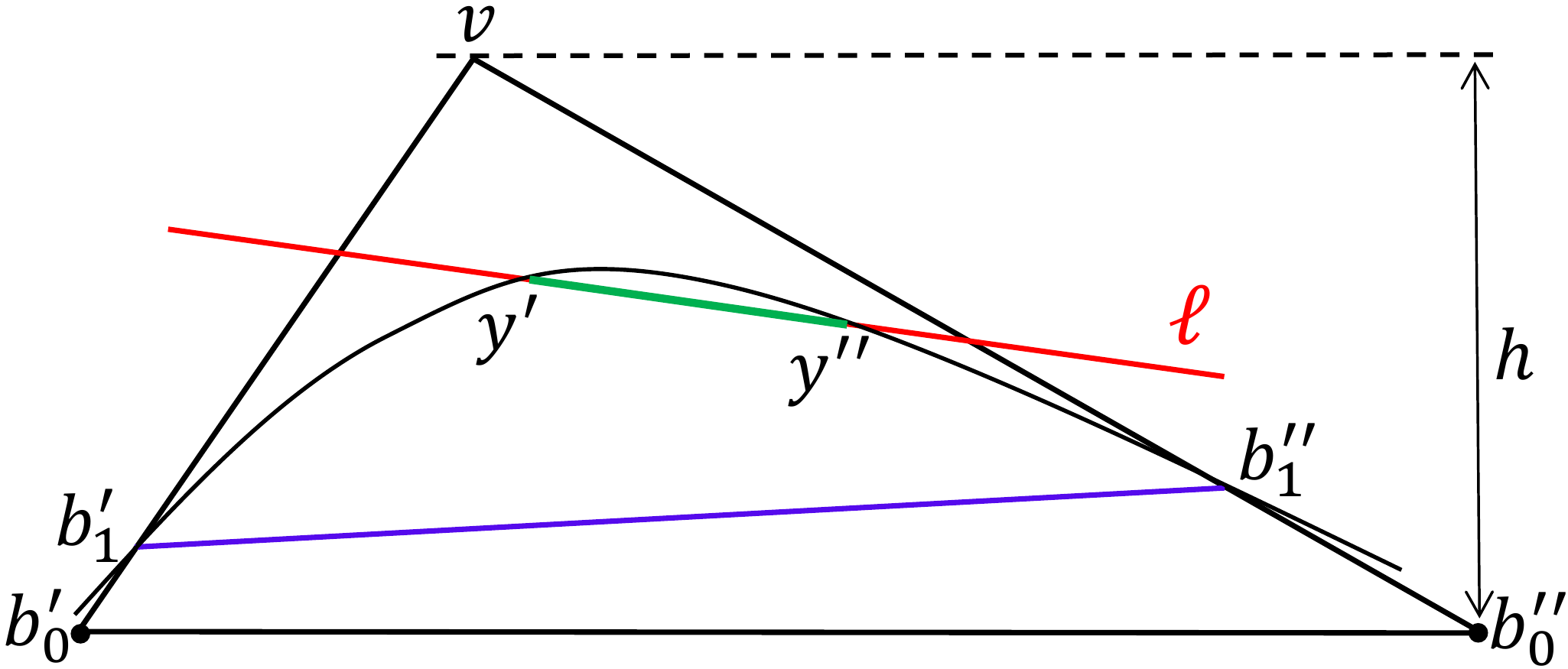}
\caption{ An illustration to Subdivision Step in $\bigtriangleup b'_0b''_0v$.}
\label{fig:small-error-ref-poly-1}
\end{figure}

\begin{proof}
By Proposition~\ref{prop:area-pixel}, $\myerr{\bigtriangleup b'_0b''_0v}\leq Pix(\bigtriangleup b'_1b''_1v)\leq A(\bigtriangleup b'_1b''_1v)+Perim(\bigtriangleup b'_1b''_1v)/2+1$. Note that $\angle b'_1vb''_1$ is obtuse.
\begin{proposition}
\label{prop:obtuse-height}
Let $T$ be a triangle with sides ${\bf a,b}$ and ${\bf c}$. Let $\alpha$ be the angle opposite to side ${\bf a}$, and ${\bf h_a}$ be the height w.r.t.\ base ${\bf a}$ in $T$. If $\alpha\geq \pi/2$ then ${\bf h_a}\leq {\bf a}/2$.
\end{proposition}
\begin{proof}
By the cosine theorem, ${\bf a^2}={\bf b^2}+{\bf c^2}-2{\bf b}{\bf c}\cdot\cos{\alpha}\geq {\bf b^2}+{\bf c^2}\geq 2{\bf b}{\bf c}\geq4\cdot A(T)=2\cdot{\bf a}\cdot {\bf h_a}$. Thus, ${\bf h_a}\leq {\bf a}/2$, as claimed.
\end{proof}
If $h\leq 6\bigdelta n$ then by Proposition~\ref{prop:obtuse-height}, the area $A(\bigtriangleup b'_1b''_1v)\leq 3\cdot|b'_1b''_1|\bigdelta n$. Since $Perim(\bigtriangleup b'_1b''_1v)\leq 3\cdot|b'_1b''_1|$ we obtain that $A(\bigtriangleup b'_1b''_1v)+Perim(\bigtriangleup b'_1b''_1v)/2+1\leq 3\cdot|b'_1b''_1|\bigdelta n+1.5\cdot|b'_1b''_1|+1\leq 4\cdot|b'_1b''_1|\bigdelta n+2$ and the claim holds (recall that $\eps=\Omega(1/n)$). Now assume that $h>6\bigdelta n$ and no line in $L_{\hat\varphi}$ with a black reference point intersects the line segments $b'_1v$ and $b''_1v$ in $\bigtriangleup b'_0b''_0v$.

\begin{figure}
\centering
\includegraphics[width=0.7\linewidth]{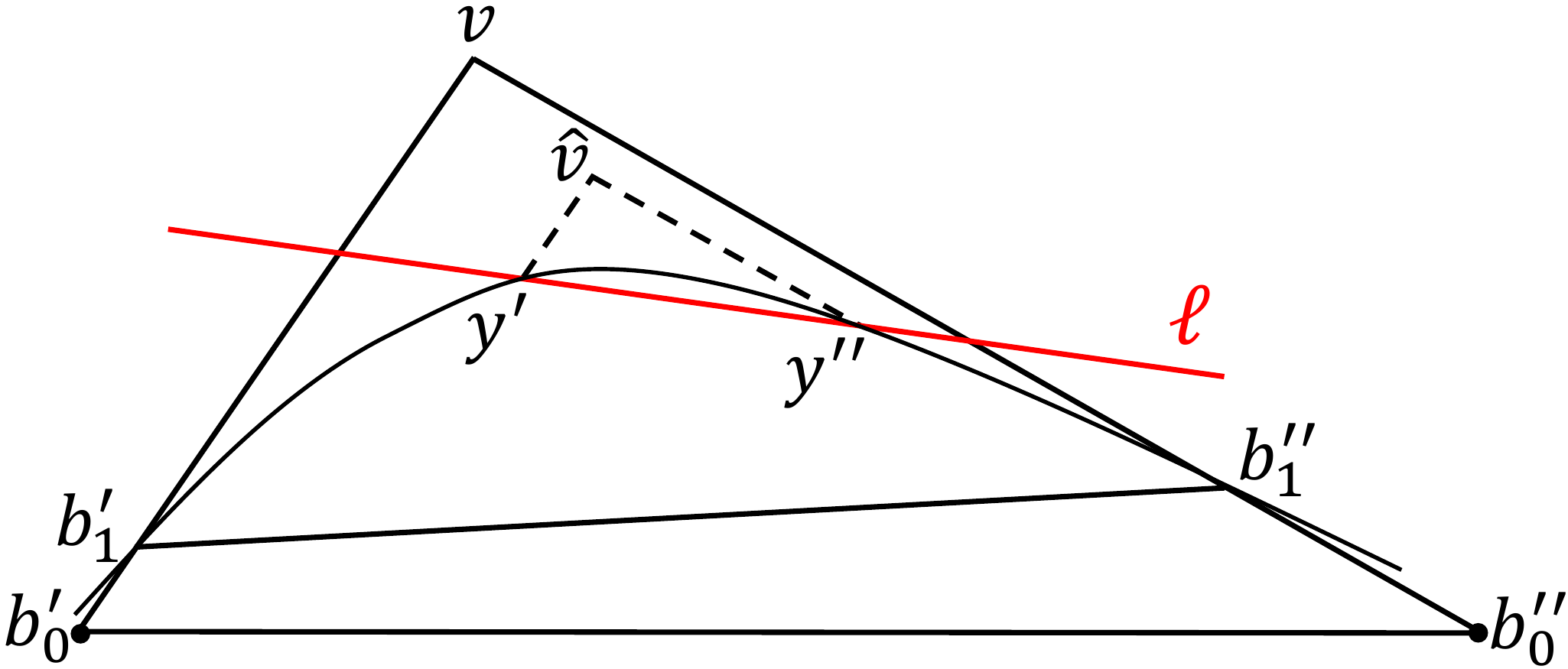}
\caption{ An illustration of triangle $\bigtriangleup y'y''\hat v$.}
\label{fig:small-error-ref-poly-2}
\end{figure}

\begin{proposition}
\label{prop:small-error-above}
Let $\bigtriangleup b'_0b''_0v$ be a triangle in which $\angle b'_0vb''_0$ is obtuse and $B_M$ intersects the sides $b'_0v$ and $b''_0v$. Let $\ell\in L_{\hat\varphi}$ be a line that intersects $B_M$ at $y'$ and $y''$, and it intersects $b'_0v$ and $b''_0v$ at $v'$ and $v''$, respectively. See Figure~\ref{fig:small-error-ref-poly-2}. Then $\myerr{v'v''v}\leq\frac{|y'y''|^2}{4}+\frac{3|y'y''|}2+1$.
\end{proposition}
\begin{proof}
Let $\hat{v}$ be a point (inside $\bigtriangleup b'_0b''_0v$) such that $\ell(y',\hat{v})$ is parallel to $\ell(b'_0v)$ and $\ell(y'',\hat{v})$ is parallel to $\ell(b''_0v)$. Since $B_M$ is convex, the portion of $B_M$ in $\bigtriangleup v'v''v$ is entirely inside $\bigtriangleup y'y''\hat{v}$. Angle $\angle y'\hat{v}y''$ is obtuse since $\angle y'\hat{v}y''=\angle b'_0vb''_0$. Then by Proposition~\ref{prop:obtuse-height}, $A(\bigtriangleup y'y''\hat{v})\leq\frac{|y'y''|^2}{4}$. Note that $Perim(\bigtriangleup y'y''\hat{v})\leq 3|y'y''|$. Since $\myerr{v'v''v}\leq Pix(\bigtriangleup v'v''v)$ then by Proposition~\ref{prop:area-pixel}, $\myerr{v'v''v}\leq\frac{|y'y''|^2}{4}+\frac{3|y'y''|}2+1$.
\end{proof}

Let $\ell\in L_{\hat\varphi}$ be the line that does not intersect the line segment $b'_1b''_1$ and that is closest to it. Let $\ell$ intersect the line segments $b'_1v$ and $b''_1v$ at $v'$ and $v''$. Then either $\angle v'v''v\leq \angle b'_0b''_0v$ or $\angle v''v'v\leq \angle b''_0b'_0v$. W.l.o.g. assume that $\angle v'v''v\leq \angle b'_0b''_0v$.
Let $\hat{\ell}$ be the line that is parallel to $\ell$ and that passes through point $b''_0$ as shown in Figure~\ref{fig:small-error-ref-poly-3}. Let $v'_1$ be the intersection point of $\hat{\ell}$ and the line segment $b'_0v$. Denote the angle between $\hat{\ell}$ and $\ell(b'_0,b''_0)$ by $\gamma$. The distance between $\ell$ and $\hat{\ell}$ is at  most $2\bigdelta n$. Otherwise there are two distinct lines from $L_{\hat\varphi}$ that pass through the line segment $b''_0b''_1$. Since $|b''_0b''_1|\leq \bigdelta n$ the distance between the two lines is less than $\bigdelta n$, contradiction.

\begin{figure}
\centering
\includegraphics[width=0.7\linewidth]{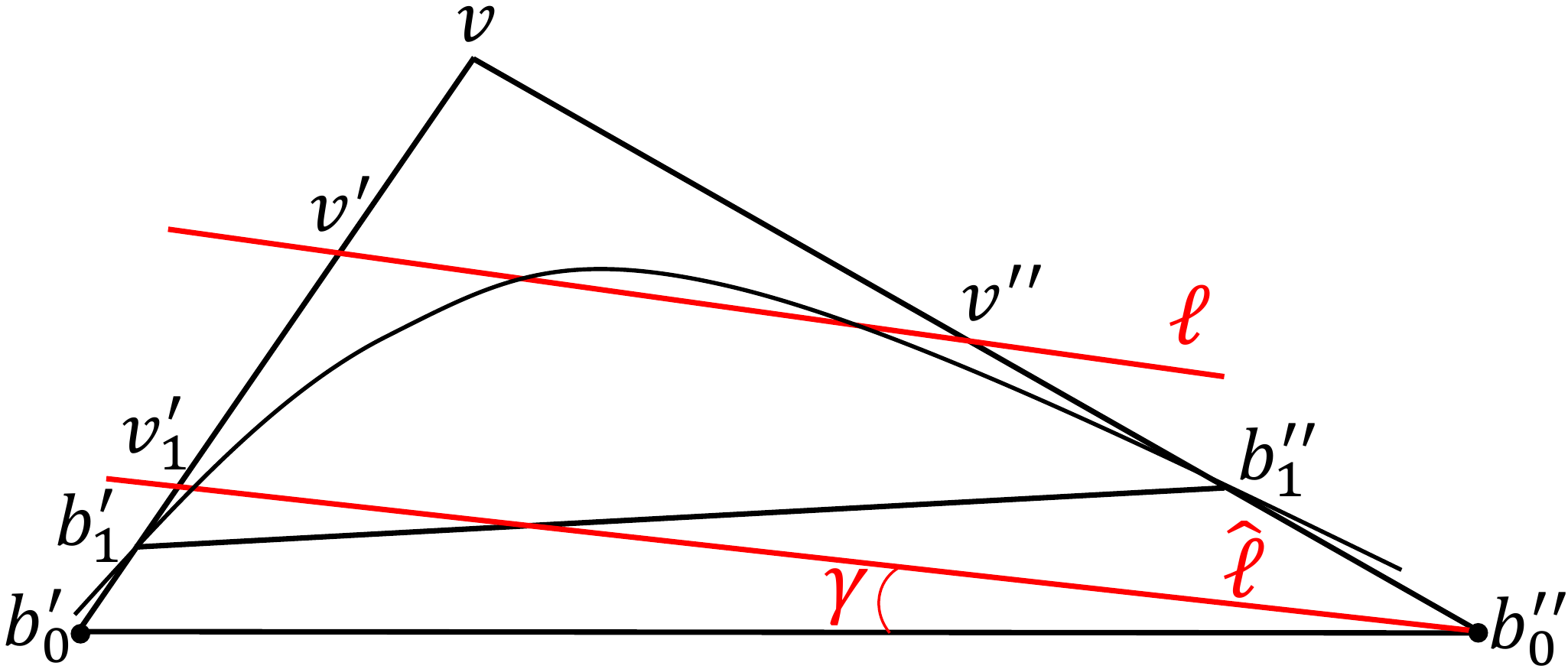}
\caption{ An illustration of line $\hat{\ell}$ in $\bigtriangleup b'_0b''_0v$.}
\label{fig:small-error-ref-poly-3}
\end{figure}

Now we find an upper on the number of black pixels in $\bigtriangleup b'_0b''_0v$. Let $B_M$ intersect $\ell$ at $y'$ and $y''$. Then $|y'y''|\leq \bigdelta n$. By Proposition~\ref{prop:small-error-above}, $Pix(\bigtriangleup v'v''v)\leq\frac{(\bigdelta n)^2}4+\frac{3\bigdelta n}2+1.$ The number of black pixels in the rectangle $b'_0b''_0v''v'$ is at most $Pix(b'_0b''_0v''v')$. The area
$$A(b'_0b''_0v''v')=A(v'_1b''_0v''v')+A(\bigtriangleup b'_0b''_0v'_1)\leq 2\bigdelta n|v'_1b''_0|+A(\bigtriangleup b'_0b''_0v'_1)\leq2|b'_0b''_0|\bigdelta n+A(\bigtriangleup b'_0b''_0v'_1).$$
The last inequality holds since $\angle b'_0v'_1b''_0$ is obtuse. Let $d_1$ (resp., $d_2$) denote the distance from the point $v'$ (resp., $v'_1$) to the line $\ell(b'_1,b''_1)$. We find an upper bound on $A(\bigtriangleup b'_0b''_0v'_1)$:

$$A(\bigtriangleup b'_0b''_0v'_1)=\frac{|b'_0b''_0|\cdot d_2}{2}=\frac{|b'_0b''_0|\cdot|b''_0v'_1|\cdot\sin{\gamma}}{2}\leq \frac{|b'_0b''_0|\cdot \sqrt{2}n(\bigdelta/2)}{2}<0.4|b'_0b''_0|\cdot\bigdelta n.$$ Thus, $A(b'_0b''_0v''v')\leq 2.4\cdot|b'_0b''_0|\bigdelta n$. The height $h\leq d_1+\bigdelta n$. By Proposition~\ref{prop:obtuse-height}, if $|b'_1b''_1|\leq10\bigdelta n$ then $d_1\leq5\bigdelta n$. It implies that $h\leq6\bigdelta n$, contradiction. Therefore, $|b'_1b''_1|>10\bigdelta n$. By the triangle inequality $|b'_0b''_0|\leq |b'_1b''_1|+2\bigdelta n$. Thus,
$$A(b'_0b''_0v''v')\leq 2.4\cdot(|b'_1b''_1|+2\bigdelta n)\bigdelta n\leq 3\cdot|b'_1b''_1|\bigdelta n.$$ The last inequality holds since $|b'_1b''_1|>10\bigdelta n$. Note that $$Perim(b'_0b''_0v''v')/2\leq 2\cdot|b'_0b''_0|\leq 2|b'_1b''_1|+4\bigdelta n.$$ Thus, by Proposition~\ref{prop:area-pixel}, $$Pix(b'_0b''_0v''v')\leq 3\cdot|b'_1b''_1|\bigdelta n+2|b'_1b''_1|+4\bigdelta n+1$$
and $$\myerr{\bigtriangleup b'_0b''_0v}\leq Pix(b'_0b''_0v''v')+Pix(\bigtriangleup v'v''v)\leq 3\cdot|b'_1b''_1|\bigdelta n+2|b'_1b''_1|+4\bigdelta n+1+\frac{(\bigdelta n)^2}4+\frac{3\bigdelta n}2+1\leq$$ $$\leq4\cdot|b'_1b''_1|\bigdelta n+2.$$
The last inequality holds since $|b'_1b''_1|>10\bigdelta n$. This completes the proof of Claim~\ref{cl:error-in-triangles}.
\end{proof}
\begin{claim}
\label{cl:small-area-above-line}
Let triangle $\bigtriangleup b'_0b''_0v$ and line $\ell$ be as defined in Step 2 of the recursive construction of $M'$. Let $v'$ and $v''$ denote the points of intersection of $\ell$ and $b'_0v$ and $b''_0v$, respectively. Let $y'$ and $y''$ be the points of intersection of $B_M$ and $\ell$. Then $\myerr{\bigtriangleup v'v''v}\leq 4\cdot|y'y''|\bigdelta n+2$.
\end{claim}
\begin{figure}[ht]
\centering
\includegraphics[width=0.7\linewidth]{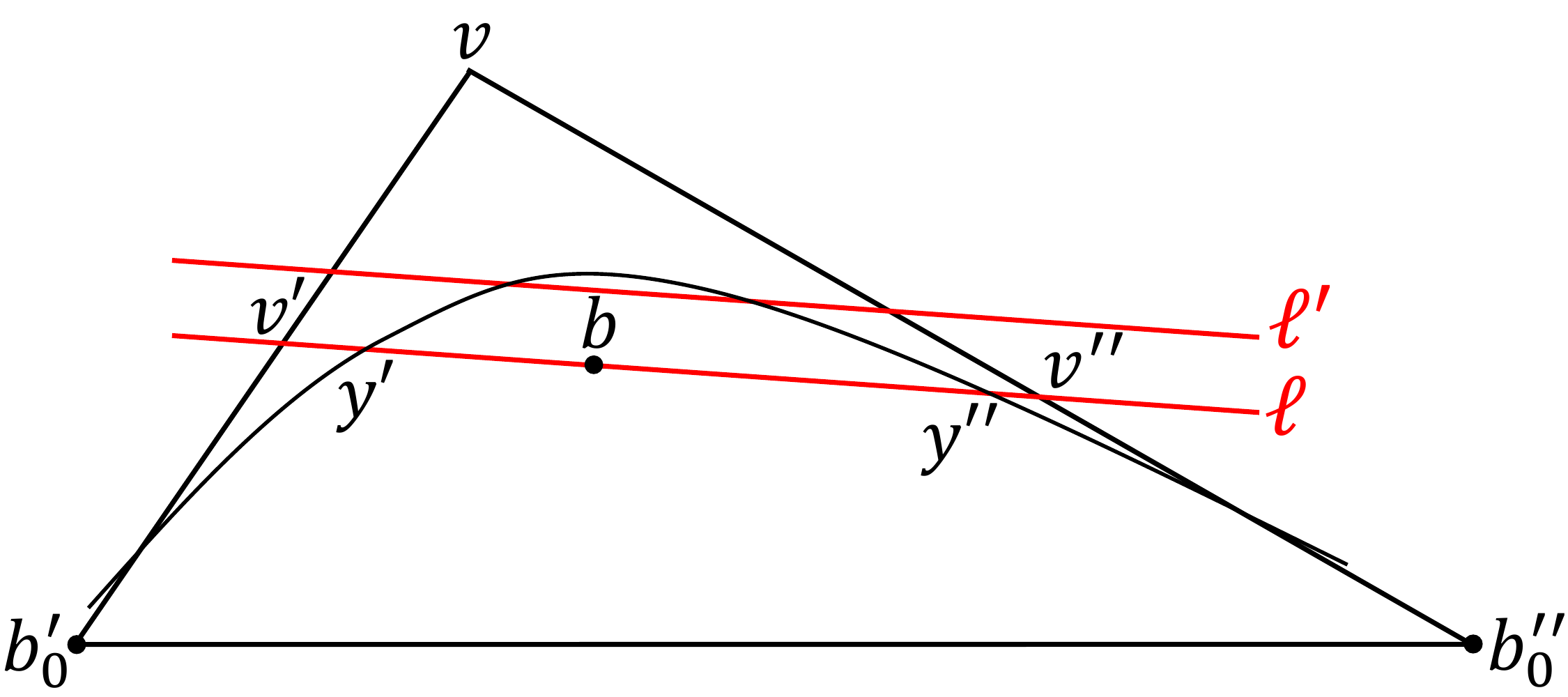}
\caption{ An illustration of line $\ell'$ in triangle $\bigtriangleup b'_0b''_0v$.}
\label{fig:error-above-l}
\end{figure}

\begin{proof}
If $|y'y''|\leq \bigdelta n$ then, by Proposition~\ref{prop:small-error-above}, $\myerr{\bigtriangleup v'v''v}\leq \frac{|y'y''|^2}4+\frac{3|y'y''|}2+1\leq 4\cdot|y'y''|\bigdelta n+2.$ Now assume that $|y'y''|>\bigdelta n$. Let $\ell'\in L_{\hat\varphi}$ be the line at distance $\bigdelta n$ from $\ell$ closer to $v$, as in Figure~\ref{fig:error-above-l}. Let $\ell'$ intersect $b'_0v$ and $b''_0v$ at $v'_1$ and $v''_1$, respectively. Then $\myerr{v'v'_1v''_1v''}$ is at most the number of black pixels in $v'v'_1v''_1v''$. Note that all black pixels in $v'v'_1v''_1v''$ are inside a rectangle with length $|y'y''|$. Thus, by Proposition~\ref{prop:area-pixel}, the number of black pixels in $v'v'_1v''_1v''$ is at most $|y'y''|\bigdelta n+2|y'y''|+1$. The distance between the points of intersection of $B_M$ with $\ell'$ is at most $\bigdelta n$. Thus, by Proposition~\ref{prop:small-error-above},
$$\myerr{v'v''v}\leq|y'y''|\bigdelta n+2|y'y''|+1+\frac{(\bigdelta n)^2}{4}+\frac{3\bigdelta n}2+1\leq 4\cdot|y'y''|\bigdelta n+2.$$ The last inequality holds since $|y'y''|>\bigdelta n$. This completes the proof of Claim~\ref{cl:small-area-above-line}.
\end{proof}

Observe that all points in $B$ lie on the boundary of a convex polygon. Images $M$ and $M'$ differ only on pixels outside of the reference box and inside the triangles $\bigtriangleup b'_1b''_1v$ and $\bigtriangleup v'v''v$. 
All the line segments in $F_1\cup F_2$ are the sides of a convex polygon which is inside an $n\times n$ square. Thus, the sum of the lengths of the line segments in $F_1\cup F_2$ is at most $4n$. Now we find an upper bound on $|\Tfin|$. Note that in the process of constructing a reference polygon starting from triangles in $\Tstart$, every triangle is subdivided into at most two new triangles. Fix a triangle $T\in\Tstart$. Consider a binary tree $\mathcal{B}_T$ rooted at $T$, where every node is some triangle obtained during the reference polygon construction and every triangle in $\mathcal{B}_T$ has at most two children triangles obtained after subdivision of their parent (during the construction). Triangles in $\Tfin$ correspond to the leaves of the binary tree. Thus, to upperbound $|\Tfin|$ we need to find the maximum possible height of $\mathcal{B}_T$ and we need to assume that the tree is full. Recall $\bigtriangleup b_0'b_0''v$ and $h$ from the construction of a reference polygon. Triangle $\bigtriangleup b_0'b_0''v$ is not subdivided if $h\leq6\bigdelta n$. By Proposition~\ref{prop:obtuse-height}, if $A(\bigtriangleup b_0'b_0''v)\leq 36(\bigdelta n)^2$ then $h\leq 6\bigdelta n$. Thus, a triangle is not subdivided if its area drops below $36(\bigdelta n)^2$. Note that every triangle in $\Tstart$ has area at most $n^2$. Consider a triangle $T_1$ in $\mathcal{B}_T$ with two children $T'_1$ and $T''_1$. Let $k$ be the height of $\mathcal{B}_T$. By Claim~\ref{claim:sum-of-roots-of-areas}, $\max\{A(T'_1),A(T'_2)\}\leq \frac 2 3 A(T_1)$.
Thus, every triangle in level $i\in[k]$ of $\mathcal{B}_T$ has area at most $(2/3)^i n^2$. The area of every triangle in level $k-1$ of $\mathcal{B}_T$ is at least $36(\bigdelta n)^2$ (otherwise, non of the triangles in this level is subdivided and the height of $\mathcal{B}_T$ cannot be $k$). We obtain that $(2/3)^{k-1} n^2\geq 36(\bigdelta n)^2$ and thus, $k\leq 5\cdot\ln \frac {30} {\eps}$. Therefore, the number of leaves in $\mathcal{B}_T$ is at most $2^k\leq n/4$ (recall that $\eps>n^{-1/5}$) and $|\Tfin|\leq 4\cdot(n/4)=n$. By Claims~\ref{cl:error-in-strips},~\ref{cl:error-in-triangles} and \ref{cl:small-area-above-line},
$$
\Dis(M,M')\leq (\sum\nolimits_{ b'_1b''_1 \in F_2}|b'_1b''_1|+\cdot\sum\nolimits_{ y'y'' \in F_1}|y'y''|)\cdot4\bigdelta n+12\bigdelta n^2+2|\Tfin|\leq 26\bigdelta n^2+2n\leq 27\bigdelta n^2.
$$
This completes the proof of Lemma~\ref{lem:nearby-reference-polygon}. }
\end{proof}


\subsection{Error analysis}\label{sec:convexity-dist-approx-error-analysis}

\begin{lemma}\label{lem:sum-of-roots-of-areas}
For each set $\Tend$ obtained in the construction of a reference polygon in Definition~\ref{def:reference-polygons},
$$\sum_{T\in\Tend} \sqrt{A(T)}<11n.$$
\end{lemma}
\begin{figure}[ht]
\centering
\includegraphics[width=0.7\linewidth]{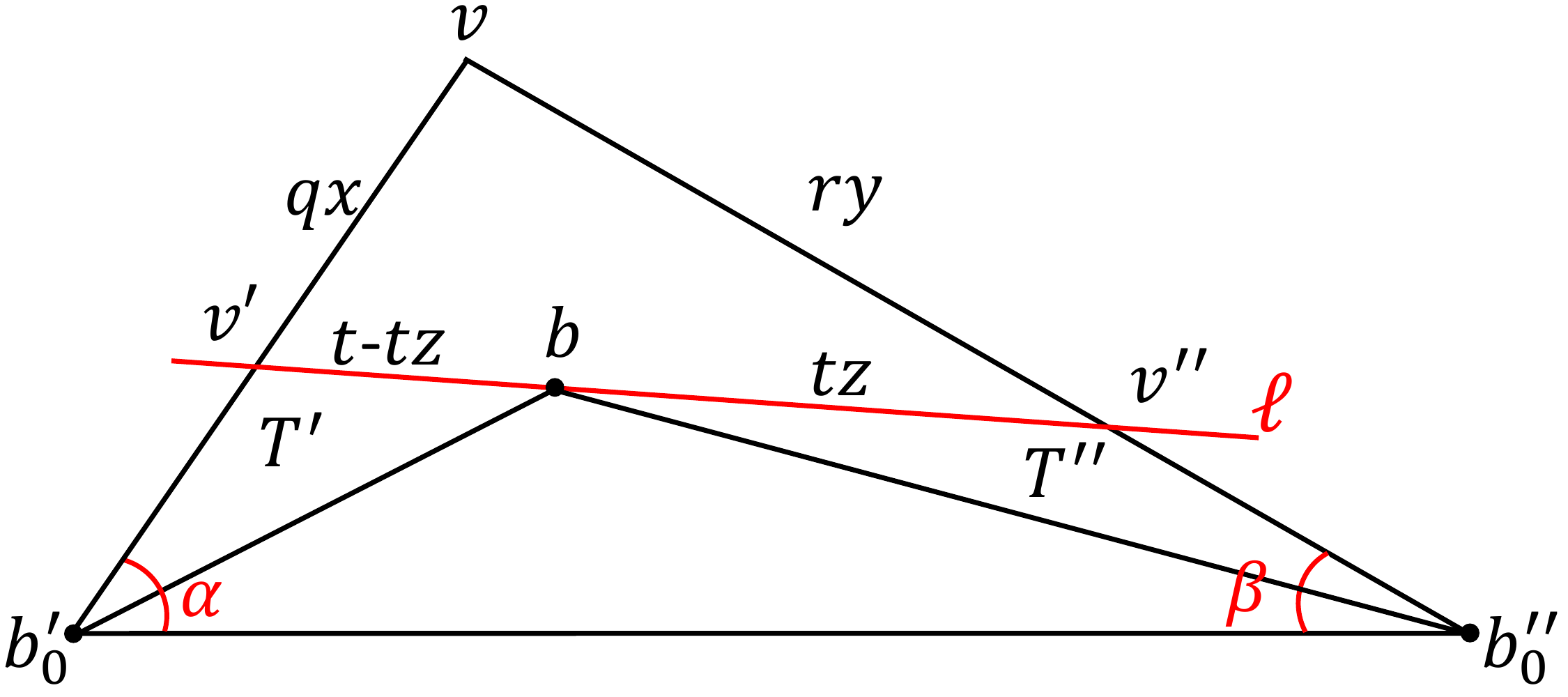}
\caption{ Triangle $\bigtriangleup b'_0b''_0v$.}
\label{fig:sum-of-roots}
\end{figure}

\begin{proof}
All triangles in $\Tend$ 
are obtained by partitioning the four initial triangles in $\Tstart$. The following claim analyzes how the area is affected by one step of partitioning.
\begin{claim}\label{claim:sum-of-roots-of-areas}
Let $T'$ and $T''$ be two gray triangles obtained from a triangle $T$ in Subdivision Step of Definition~\ref{def:reference-polygons}. Then
$\sqrt{A(T')}+\sqrt{A(T'')}\leq \sqrt{\frac 2 3 \cdot A(T)}.$
\end{claim}

\begin{proof}
Observe that $\sqrt{A(T')}+\sqrt{A(T'')}$ is maximized when $b'_0=b'$ and $b''_0=b''$. W.l.o.g.\ we prove the lemma for this case. We use notation from Figure~\ref{fig:sum-of-roots}. Recall that a triangle $T$ is partitioned only if its height $h\geq 6\bigdelta n.$ Since the sides of $T$ are of length at most $\sqrt 2 n$, the height is that large only if both angles adjacent to the base $b'_0b''_0$ are greater than $4\bigdelta$. (To see this, consider an angle $\alpha$ between the base and a side of length $a$. We get
$6\bigdelta n\leq h=a\cdot\sin\alpha\leq \sqrt 2 n \cdot \alpha$.
Thus, $\alpha\geq 6\bigdelta/\sqrt 2> 4\bigdelta.$)

First, we find the maximum value of $\sqrt{A(T')}+\sqrt{A(T'')}$ for a fixed line $\ell$ on which position of point $b$ varies. 
\mch{Let $\alpha=\angle b''_0b'_0v$, $\beta=\angle b''_0b'_0v$ and $\gamma$ be the angle between lines $\ell$ and $\ell(b'_0,b''_0)$. W.l.o.g.\ assume that $\angle v'v''v\leq \beta$. Then $\angle v''v'v=\alpha+\gamma$ and $\angle v'v''v=\beta-\gamma$. By the construction of triangles in $\Tend$, $\alpha+\beta\leq\frac{\pi}{2}$ and $\gamma\leq\frac \bigdelta 2$. Let $q=|b'_0v|$, $r=|b''_0v|$, $t=|v'v''|$ and $qx=|v'v|$, $ry=|v''v|$, $tz=|bv''|$ ($x,y,z\in [0,1]$).}  Let\mch{
$$
f(z)=\sqrt{A(T')}+\sqrt{A(T'')}=\sqrt{\frac{q(1-x)\cdot t(1-z)\sin(\alpha+\gamma)}{2}}+\sqrt{\frac{r(1-y)\cdot tz\sin(\beta-\gamma)}{2}}.
$$
Thus, $f(z)=\sqrt{C_1\cdot(1-z)}+\sqrt{C_2\cdot z}$, where $C_1=A(\bigtriangleup b'_0v'v'')$, $C_2=A(\bigtriangleup b''_0v''v')$ are constants. By the  Cauchy-Schwarz inequality, $f(z)=\sqrt{C_1\cdot(1-z)}+\sqrt{C_2\cdot z}\leq\sqrt{(C_1+C_2)(1-z+z)}=\sqrt{C_1+C_2}$.

Next, we find the maximum value of $C_1+C_2$ varying position of $\ell$ inside $T$. We use the fact that
$$
C_1=A(\bigtriangleup b'_0v''v)-A(\bigtriangleup v'v''v)=\frac{(q-qx)ry\cdot\sin(\alpha+\beta)}{2},
$$ $$
C_2=A(\bigtriangleup b''_0v'v)-A(\bigtriangleup v'v''v)=\frac{(r-ry)qx\cdot\sin(\alpha+\beta)}{2}
$$ to obtain
$$
C_1+C_2=\frac{(x+y-2xy)qr\cdot\sin(\alpha+\beta)}{2}=(x+y-2xy)A(T).
$$
We need to show that $x+y-2xy\leq 2/3$. Let $\hat c=\frac y x=\frac{\sin\beta\sin(\alpha+\gamma)}{\sin\alpha\sin(\beta-\gamma)}$. Since $\hat c$ is constant and the geometric mean of two numbers is at most their arithmetic mean $$\sqrt{x+y-2xy}=\sqrt{2\hat c} \cdot \sqrt{x(\frac{1+\hat c}{2\hat c}-x)}\leq \sqrt{2\hat c}\cdot \frac 1 2\cdot (x+\frac{1+\hat c}{2\hat c}-x)=\frac{1+\hat c}{\sqrt{8\hat c}}.$$ We prove that $\frac{(1+\hat c)^2}{8\hat c}\leq \frac 2 3$ which is equivalent to $(3\hat c-1)(\hat c-3)\leq0$. The latter inequality holds if $1\leq\hat c\leq3.$ Function $\sin\theta$ is increasing on $[0,\pi/2]$ thus, $1\leq\hat c.$ Now we show that $\hat c\leq3.$ If $\gamma=0$ the inequality holds. Let us assume that $\gamma>0.$
We need to prove that
$$
\frac{\sin\beta\sin(\alpha+\gamma)}{\sin\alpha\sin(\beta-\gamma)}=\frac{\cot\gamma+\cot\alpha}{\cot\gamma-\cot\beta}\leq3.
$$
Function $\cot\theta$ is decreasing on $(0,\pi/2]$ thus,
$$
\frac{\cot\gamma+\cot\alpha}{\cot\gamma-\cot\beta}\leq\frac{\cot\gamma+\cot4\bigdelta}{\cot\gamma-\cot4\bigdelta}\leq\frac{\cot\gamma+\cot\bigdelta}{\cot\gamma-\cot\bigdelta}\leq3.
$$
The last inequality is equivalent to $2\cot\bigdelta\leq\cot\gamma$ which is true since $2\cot\bigdelta\leq\cot\frac \bigdelta 2\leq \cot\gamma$. This completes the proof of Claim~\ref{claim:sum-of-roots-of-areas}}\end{proof}

Let $A_1,\dots,A_4$ be the areas of the first four triangles in $\Tstart$. Then $\sum_{i=1}^4 A_i\leq n^2$.
By construction of triangles in $\Tend$, Claim~\ref{claim:sum-of-roots-of-areas}, and concavity of the square root function,
$$
\sum_{T\in \Tend} \sqrt{A(T)}\leq
K\cdot\sum_{j=1}^4 \sqrt{A_j}\leq 2K\sqrt{A_1+A_2+A_3+A_4}\leq 2K\cdot n,
$$
where $K=\sum_{m=0}^{\infty}(\sqrt{2/3})^m=(1-\sqrt{2/3})^{-1}<5.5$.
%
This completes the proof of Lemma~\ref{lem:sum-of-roots-of-areas}.
\end{proof}

Let $M$ be an input image, $S$ be the set of samples obtained by the algorithm, and $s$ be the parameter in the algorithm.
For any image $M'$, let $d(M')=\dis(M,M')$ and
$\dout(M')=\frac 1 s \cdot |\{u\in S : M[u]\neq M'[u]\}|.$
Also, for any region $R\subseteq\domain$, let $d(M'|_R)=\frac 1 {n^2} \cdot |\{u\in R : M[u]\neq M'[u]\}|$
and $\dout(M'|_R)=\frac 1 s \cdot |\{u\in S\cup R : M[u]\neq M'[u]\}|.$
\fi

\begin{lemma}\label{lem:accuracy-on-ref-polygons}
With probability at least $2/3$ over the choice of the samples taken by Algorithm~\ref{alg:convexity-dist-approximation},
$|\dout(M')-\dis(M,M')|\leq 5\mydelta/6$ for all reference polygons $M'$.
\end{lemma}
\begin{proof}

Consider a region $R=(R_+,R_-),$ partitioned into two regions $R_+$ and $R_-,$ such that in some step of the algorithm we are checking the assumption that $R_+$ is black and $R_-$ is white, i.e., evaluating $\dout_+(R_+)+\dout_-(R_-).$ Let $\stripset$ be the set of all such regions $R$. We will show that with probability at least 2/3, the estimates $\dout_+(R_+)+\dout_-(R_-)$ are accurate on all regions in $\stripset$.

Fix $R=(R_+,R_-)\in\stripset$. Let $\Gamma$ be the set of misclassified pixels in $R$, i.e., pixels in $R_+$ which are white in $M$ and pixels in $R_-$ which are black in $M$. Define $\gamma=|\Gamma|/n^2$. Algorithm~\ref{alg:convexity-dist-approximation} approximates $\gamma$ by $\dout_+(R_+)+\dout_-(R_-)=\frac 1 s |\Gamma\cap S|.$ Equivalently, it uses the estimate $\frac 1 p |\Gamma\cap S|$ for $|\Gamma|$ (recall that $p=s/n^2$). The error of the estimate is $\err(R)=\frac 1 p |\Gamma\cap S|-|\Gamma|$.

\begin{claim}\label{claim:error-for-R}
$\Pr[|\err(R)|> \sqrt\gamma\cdot c\mydelta n^2]\leq 2\exp(-\frac 3 8 c^2\mydelta^2 s)$, where $c=1/21$.
\end{claim}

\begin{proof}
For each pixel $u,$ we define random variables $\chi_u$ and $X_u$,
where $\chi_u$ is the indicator random variable for the event $u\in S$ (i.e., a Bernoulli variable with the probability parameter $p$), whereas $X_u=\frac{\chi_u}p-1$.
Then our estimate of $|\Gamma|$ is $\frac 1 p |\Gamma\cap S|=\frac 1 p\sum_{u\in\Gamma}\chi_u,$ whereas
$\err(R)=\sum_{u\in\Gamma}X_u$. We use Bernstein inequality
\ifnum\full=1
(Theorem~\ref{thm:bernstein})
\else
(stated in the {\color{black} full version} for completeness)
\fi
with parameters
$m=\gamma n^2$ and $z=\sqrt\gamma\cdot c\mydelta n^2$ to bound $\Pr[\sum_{u\in\Gamma}X_u> \sqrt\gamma\cdot c\mydelta n^2]$.
The variables $X_u$ are identically distributed. The maximum value of $|X_u|$ is $a=\frac{1-p}{p}$. Note that
$\E[X_u^2]=\frac 1 {p^2}\E[(\chi_u -p)^2]=\frac 1 {p^2}\Var[\chi_u]=\frac{1-p}{p}=a$.
We assume w.l.o.g.\ that $z<|\Gamma|.$ (If $z\geq |\Gamma|$ then $\sum_{u\in\Gamma}X_u$ can never exceed $z$, and the probability we are bounding is 0.)
By Bernstein inequality,
\begin{eqnarray*}
\Pr\left[\sum_{u\in\Gamma}X_u>z\right]&\le& \exp \left(\frac{-z^2/2}{a|\Gamma|+a\cdot z/3}\right)
<\exp\left(-\,\frac 3 8 \cdot \frac{z^2\cdot p}{|\Gamma|}\right)=\exp\left(-\,\frac 3 8 \cdot \frac{\gamma\cdot c^2\mydelta^2n^4}{\gamma n^2}\cdot \frac s {n^2}\right)\\
&=&\exp(-\frac 3 8 c^2\mydelta^2 s).
\end{eqnarray*}

The second inequality holds because $a< 1/p$ and $z<|\Gamma|$.
The equalities are obtained by substituting the expressions for $z,|\Gamma|,$ and $p$, and simplifying.
By symmetry, $\Pr[|\err(R)|\geq z]\leq 2\exp(-\frac 3 8 c^2\mydelta^2 s)$.
\end{proof}
\ifnum\full=0
The rest of the proof appears in the {\color{black} full version}.
\end{proof}
\else
\begin{claim}\label{claim:convexity-approx-regions-count}
The number of regions in $\stripset$ is at most $50/\bigdelta^8$.
\end{claim}
\begin{proof}
Let $k=1/\bigdelta$. There are four types of regions in $\stripset$, each corresponding to a different call of the form $\dout_+(R_+)+\dout_-(R_-)$ in the algorithm. The first type is a horizontal double strip of the form $R_+=\emptyset$ and $R_-= W_{\ell_0}\cup W_{\ell_2}$. There are ${k+1}\choose 2$ such strips.
The second type is where $R_+$ is a black triangle $\bigtriangleup b_0b_1b_2$ (or $\bigtriangleup b_0b_3b_2$) and $R_-$ is a vertical strip $W_{\ell_1}$ (respectively, $W_{\ell_3}$). For each horizontal double strip, there are $2k-1$ vertical strips. For each of them, there are $k$ ways to choose a reference point on the vertical line that delineates the strip. So, overall, there are $(k+1)k(k-1/2)k$ regions of type 2. Type 1 and 2 together have at most $.5k^8$ regions.
%
%
Regions of type 3 are black quadrilaterals of the form $R_+=b'_0b_0''b'b''.$ Each quadrilateral is defined by two reference lines, $b'b_0'$ and $b''b_0'',$ with two reference points on each. There are ${\pi k} \choose 2$ ways to choose reference directions for the two lines. For each of them, there are at most $\sqrt 2 k\cdot {{\sqrt 2 k}\choose 2}$ ways to choose a reference line and two reference points. Overall, the number of quadrilaterals in $\stripset$ is at most $\pi^2 k^8$.
Finally, regions of type 4 are contained in triangles of the form $\bigtriangleup v b'_0 b''_0$; they are of the form either $R_+=\emptyset, R_-=\bigtriangleup v b'_0 b''_0$ or $R_+=\bigtriangleup bb'_0b''_0$, $R_-=\bigtriangleup vv'v''$. In the former case, regions are defined by two line-point pairs $(\ell(b'_0,v),b'_0)$ and $(\ell(b''_0,v),b''_0)$. There are ${\pi k} \choose 2$ pairs of reference directions. For each of them, there are at most $\sqrt 2 k$ choices for each reference line and $\sqrt 2 k$ choices for each reference points. In the latter case, they are defined by three reference line-point pairs: $(\ell(b'_0,v),b'_0),(\ell(b''_0,v),b''_0),$ and $(\ell(v',v''),b),$ but the direction of the line through $v'v''$ is determined by $b'_0,b''_0$. As before, there are ${\pi k} \choose 2$ pairs of reference directions. For each of them, there are at most $\sqrt 2 k$ choices for each reference line and $\sqrt 2k$ choices for each reference points. Overall, the number of regions of type 4 is upper-bounded by $4\pi^2 k^8$.
Overall, $|\stripset|\leq (5\pi^2+.5) k^8<50 k^8=50/\bigdelta^8,$ as claimed.
\end{proof}

By taking a union bound over all regions in $\stripset$ and applying Claims~\ref{claim:error-for-R}--\ref{claim:convexity-approx-regions-count}, we get that the probability that for one or more of them the error is larger than stated in Claim~\ref{claim:error-for-R} is at most $|\stripset|\cdot 2\exp(-\frac 3 8 c^2\mydelta^2 s)
\leq\frac{100}{\bigdelta^8}\cdot\exp(-\frac 3 8 c^2\mydelta^2 s)\leq 1/3$,
where the last inequality holds provided that $s\geq C\frac 1{\mydelta^2}\ln\frac 1\mydelta$ for some sufficiently large constant $C$. We get that
\begin{align}\label{eq:low-error-on-all}
\Pr[|\err(R)|\leq \sqrt\gamma\cdot c\mydelta n^2 \text{ for all $R\in\stripset$}]\geq 2/3.
\end{align}

Now suppose that event in (\ref{eq:low-error-on-all}) holds, that is, the error is low for all regions. Fix a reference polygon $M'$. Consider the partition of $M'$ into regions from $R=(R_+,R_-)\in\stripset$ on which Algorithm~\ref{alg:convexity-dist-approximation} evaluates $\dout_+(R_+)+\dout_-(R_-)$ while implicitly computing $\dout(M')$. Let $\stripset_{M'}\subset \stripset$ be the set of regions in the partition. Recall the four types of regions from the proof of Claim~\ref{claim:convexity-approx-regions-count}. Then  $\stripset_{M'}$ contains one region of type 1 and two regions of type 2, defined by the reference box of $M'$. Denote their areas by $A'_1,A'_2,A'_3$. For each triangle $T\in\Tend$ created during the construction of $M'$ in Definition~\ref{def:reference-polygons}, the set $\stripset_{M'}$ contains at most one region of type 3 and at most one region of type 4. They were implicitly colored, respectively, in Item~\ref{item:ref-poly-definition} and Items~\ref{item:ref-poly-definition-tall}-\ref{item:ref-poly-definition-short} of Definition~\ref{def:reference-polygons}, when triangle $T$ was processed. Let $A_T$ and $A'_T$ denote their respective areas.

Recall that $A(R)$ denotes the area of $R$ and that an approximate (but precise enough for asymptotic analysis) upper bound on the number of misclassified pixels in $R$ is $A(R)$. 
Since the event in (\ref{eq:low-error-on-all}) holds,
$$\err(M')\leq \sum_{R\in\stripset_{M'}}\err(R)
\leq c\mydelta n \sum_{R\in\stripset_{M'}} \sqrt{A(R)}
\leq c\mydelta n \big(\sum_{j=1}^3\sqrt{A'_j}  + \sum_{T\in\Tend}(\sqrt{A_T}+\sqrt{A'_T})\big).
$$
Since $\sum_{j=1}^3A'_j\leq n^2$ and $A_T+A'_T\leq A(T)$ for all $T\in\Tend$,  by concavity of the square root function,
$$\sum_{j=1}^3\sqrt{A'_j}\leq\sqrt{3\sum_{j=1}^3 A'_j}\leq \sqrt{3} n
\text{ and }
\sqrt{A_T}+\sqrt{A'_T}\leq\sqrt{2(A_T+A'_T)}\leq\sqrt{2A(T)}.$$
We substitute these expressions in the previous inequality, use Lemma~\ref{lem:sum-of-roots-of-areas} and recall that $c=1/21$:
$$\err(M')\leq
c\mydelta n \big(\sqrt{3} n  + \sqrt 2 \sum_{T\in\Tend}\sqrt{A(T)}\big)
\leq
c\mydelta n^2 (\sqrt{3}  + 11\sqrt 2)\leq \frac 5 6 \mydelta n^2.
$$
This holds for all reference polygons $M'$ as long as the event in (\ref{eq:low-error-on-all}) happens, i.e., with probability at least 2/3. This completes the proof of Lemma~\ref{lem:accuracy-on-ref-polygons}.
\end{proof}

For completeness, we state Bernstein's inequality, which was used in the proof of Lemma~\ref{lem:accuracy-on-ref-polygons}.

\begin{theorem}[Bernstein's inequality]\label{thm:bernstein}
 Let $X_1,\dots,X_m$ be $m$ independent zero-mean random variables, where $|X_i|\leq a$ for all $i\in [m]$. Then for all positive $z$,
$$
\Pr\left[\sum_{i = 1}^m X_i > z\right] \le \exp\left(- \frac{z^2/2 }{\sum_{i = 1}^m \E[X_i^2]+a\cdot z/3}\right) .
$$

\end{theorem}

\subsection{Wrapping up: proof of Theorem~\ref{thm:convexity_dist_appr} and corollary on agnostic learning}\label{sec:convexity-dist-approx-wrap-up}
\subparagraph{Analysis of Algorithm~\ref{alg:convexity-dist-approximation}.}
Let $d_M$ be the distance of $M$ to convexity. Then there exists a convex image $M^*$ such that $\dis(M,M^*)=d_M$.
By Lemma~\ref{lem:nearby-reference-polygon}, there is a reference polygon $\hat M$ such that $\dis(M^*,\hat M)\leq \mydelta/6$, and consequently, $d_M\leq\dis(M,\hat M)\leq d_M+\mydelta/6$.
By Lemma~\ref{lem:accuracy-on-ref-polygons},
with probability at least 2/3 over the choice of the samples taken by Algorithm~\ref{alg:convexity-dist-approximation},
$|\dout(M')-\dis(M,M')|\leq 5\mydelta/6$ for all reference polygons $M'$.
Suppose this event happened. Then $\dout\geq d_M-5\mydelta/6$ because $\dis(M,M')\geq d_M$ for all convex images $M'$. Moreover,  $\dout(\hat{M})\leq\dis(M,\hat{M})+5\mydelta/6\leq d_M+\mydelta.$ Thus,
$d_M-5\mydelta/6\leq\dout\leq\dout(\hat{M})\leq d_M+\mydelta.$ That is, $|d_M-\dout|\leq \mydelta$ with probability at least 2/3, as required.

\subparagraph{Sample and time complexity of Algorithm~\ref{alg:convexity-dist-approximation}.} The number of samples taken by the algorithm is $s=O(\mydelta^{-2}\log \mydelta^{-1})$.

Next we explain how to implement it to run in time $O(\mydelta^{-8})$. Refer to Figure~\ref{fig:tolerant-reducing-area}.
Each instance triangle $\bigtriangleup b'b''v$ of the dynamic programming in subroutine \Best is specified by two line-point pairs: $(\ell(b',v),b'),(\ell(b'',v),b'')$. The number of line-point pairs is $O(\mydelta^{-3})$ because for each we select
the reference direction, the shift of the line, and the reference point, each
in $O(\mydelta^{-1})$ ways.  Hence, we have $O(\mydelta^{-6})$ entries in the dynamic
programming table for \Best.

In the process of solving an instance of  \Best, we consider $O(\mydelta^{-2})$ possibilities for points $b_0',b_0''$, that is, $O(\mydelta^{-8})$ possibilities over all instances. We show how to evaluate each of the possibilities in amortized time $O(1)$.
For that, we count white and black sample pixels in each sub-area in Figure~\ref{fig:tolerant-reducing-area} in amortized time $O(1)$.

First, we show how to do it for the entire triangle $\bigtriangleup b'b''v$.  We have $O(\mydelta^{-6})$
triangles that can be partitioned into $O(\mydelta^{-5})$ groups by specifying
the first line-point pair $(\ell(b',v),b')$ and the second line (through $b''$ and $v$). That is, within each group, we vary
only point $b''$ on the second line.  We sort all sample points $p\in S$
according to the angle of the segment $pb'$. Similarly, we sort
the reference points $b''$ on the second line according to the angle of the segment $b''b'$.  After sorting, a single scan can establish the counts of white and black pixels in the triangles. Clearly,
we can sort in time $o(\mydelta^{-3})$. Thus, we
compute white/black counts for all instance triangles of \Best in time $o(\mydelta^{-8})$.

When we consider a possibility in \Best, the triangle
$\bigtriangleup b'_0b''_0v$ is also an instance triangle, so we can find the
white/black counts for the quadrilateral $b'b_0'b_0''b''$ by computing the difference between the counts for entire triangle $\bigtriangleup b'b''v$ and triangle $\bigtriangleup b_0'b_0''v$, that is, in time $O(1)$.

When we consider a possibility in subroutine \BestFixed, we need the
counts for the four parts of $\bigtriangleup b'_0b''_0v$. Since we already calculated the counts for $\bigtriangleup b'_0b''_0v$ and because we can perform subtractions, it is
enough to do it for three parts.  Two of them,
$\bigtriangleup b'_0bv'$ and
$\bigtriangleup b''_0bv''$, are instance triangles for \Best, so we already calculated their counts. The third we choose is the triangle $\bigtriangleup v'v''v$.  Note that this triangle is specified by three reference lines, so
there are $O(\mydelta^{-6})$ such triangles. We make a table for all of them.
To fill the table, we consider $O(\mydelta^{-5})$ groups:  we group together triangles for which line $\ell$ has
a common direction. By sorting samples in $S$, we can compute the counts for
each group in time $o(\mydelta^{-3})$. Thus, the time for filling the second table is $o(\mydelta^{-8})$. To summarize, Algorithm~\ref{alg:convexity-dist-approximation} runs in time $O(\mydelta^{-8})$. This completes the proof of Theorem~\ref{thm:convexity_dist_appr}.
\end{proof}

\begin{corollary}\label{cor:convexity-agnostic-learner}
The class of convex images is properly agnostically PAC-learnable with sample complexity  $O(\frac{1}{\mydelta^2}\log\frac{1}{\mydelta})$ and time complexity $O(\frac{1}{\mydelta^8})$ under the uniform distribution.
\end{corollary}
\begin{proof}
We can modify Algorithm~\ref{alg:convexity-dist-approximation} to output, along with $\dout$, a reference polygon $\hat{M}$ with $\dout(\hat{M})=\dout$. With an additional DP table, we can compute which points became its vertices. By the analysis of Algorithm~\ref{alg:convexity-dist-approximation}, with probability at least 2/3, the output $\hat{M}$ satisfies $\dis(M,\hat{M})\leq d_M+\mydelta.$
\end{proof}
\fi

\section{Distance Approximation to the Nearest Connected Image}\label{sec:connectedness}
To define {\em connectedness}, we consider {\em the  image graph $G_M$} of an image $M$. The vertices of $G_M$ are
$\{(i,j)\mid M[i,j]=1\}$, and two vertices $(i,j)$ and $(i',j')$ are connected by an edge if $|i-i'|+|j-j'|=1$.
In other words, the image
graph consists of black pixels connected by the grid lines. The image is
{\em connected} if its image graph is connected.

\begin{theorem}\label{thm:connectedness_dist_appr}
There is a block-uniform $\mydelta$-additive distance approximation algorithm for connectedness with sample complexity $O(\frac{1}{\mydelta^{4}})$ and running time $\exp\left(O\left(\frac 1 \mydelta \right)\right)$.
\end{theorem}

\ifnum\full=1
\subsection{Border Connectedness}\label{sec:border_connectedness}
\fi

The first idea in our algorithms for connectedness is that we can modify an image
\ifnum\full=1
in a relatively few places by superimposing a grid on it
(as shown in Figures~\ref{fig:Hawaii} and~\ref{fig:Hawaii-grided}),
and as a result obtain an image whose distance to connectedness is determined by the properties of individual squares into which the grid lines partition the image. The squares and the relevant property of the squares  are defined next.
\begin{figure}[ht]
\begin{minipage}[b]{0.45\linewidth}
\centering
\includegraphics[width=\linewidth]{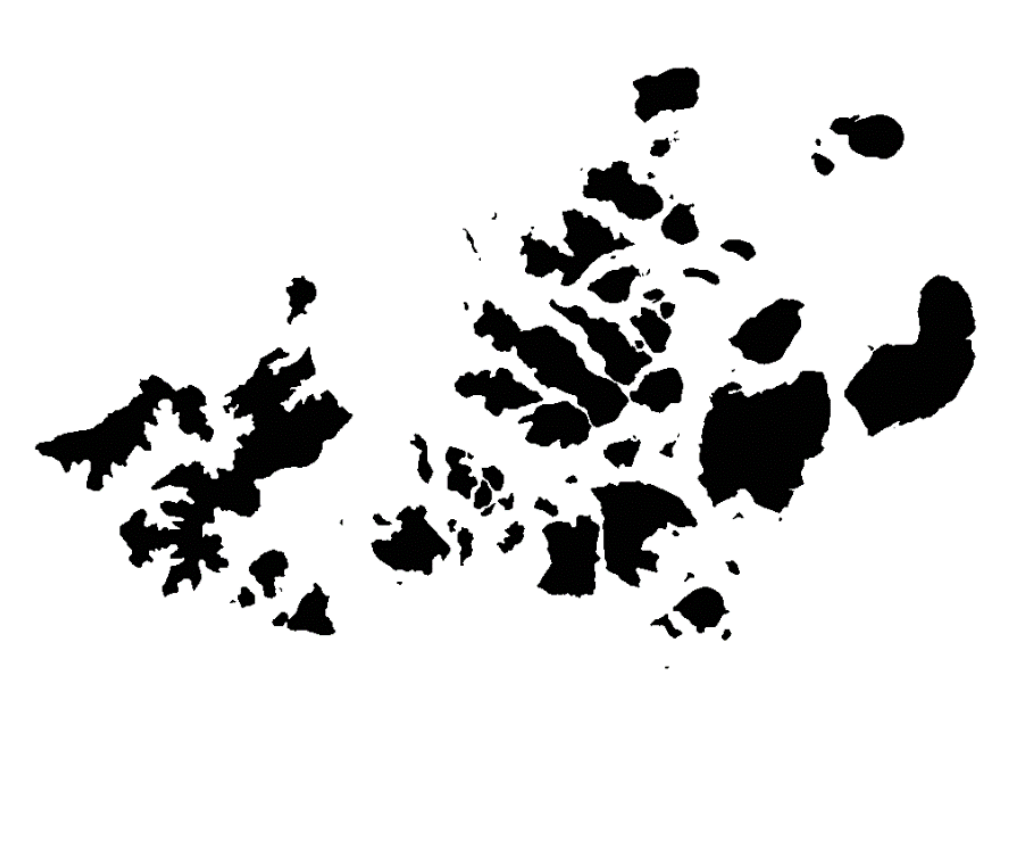}
\caption{An image $M$.}
\label{fig:Hawaii}
\end{minipage}
\hspace{0.1\linewidth}
\begin{minipage}[b]{0.45\linewidth}
\centering
\includegraphics[width=\linewidth]{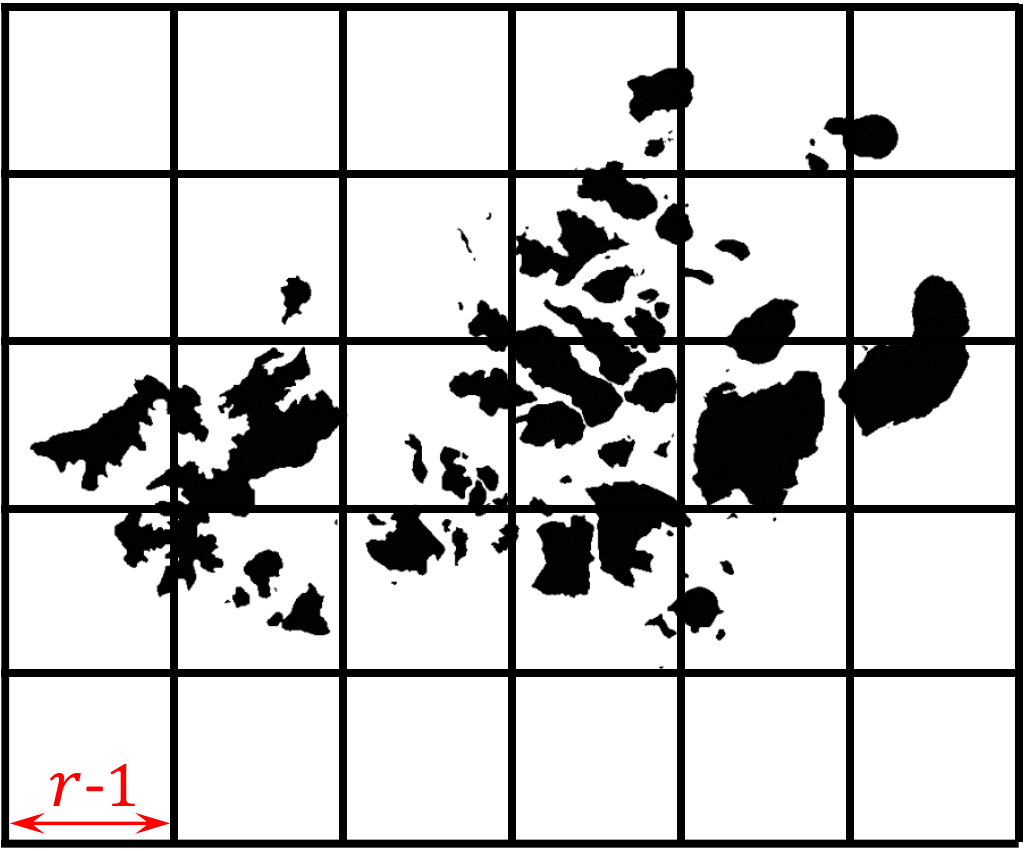}
\caption{A gridded image obtained from~$M$.}
\label{fig:Hawaii-grided}
\end{minipage}
\end{figure}
\else
by superimposing a grid on it
(see figures in the {\color{black} full version}),
and as a result obtain a nearby image whose distance to connectedness is determined by the properties of individual squares into which the grid lines partition the image. The squares and the relevant property of the squares  are defined next.
\fi

For a set $S\subset \integerset{n}^2$ and $(i,j)\in \integerset{n}^2$, we define
$S+(i,j)=\{(x+i,y+j):~(x,y)\in S\}$.

\begin{definition}[Squares and grid pixels]\label{def:squares} Fix a side length $n\equiv 1\pmod r$.
For all integers $i,j\in\integerset{n-r}$, where $i$ and $j$ are divisible by $r$, the $(r-1)\times (r-1)$ image that
consists of all pixels in $[r-1]^{2}+(i,j)$ is called an {\em $r$-square} of $M$. The set of all $r$-squares of $M$ is denoted $S_r$.

The pixels that do not lie in any squares of $S_r$, i.e., pixels $(i,j)$ where $i$ or $j$ is divisible by $r,$ are called {\em grid pixels}. The set of all grid pixels is denoted by $\gp_r$.
 \end{definition}

\begin{claim}\label{claim:GP-size}
$|\gp_r|\leq 2n^2/r$.
\end{claim}
\ifnum\full=1
\begin{proof}
$|\gp_r|=2(\frac{n-1}r +1)n-(\frac{n-1}r +1)^2
\leq 2n^2/r$.
\end{proof}
\fi
Note that a square consists of pixels of an \rblock, with the pixels of the first row and column removed. Therefore, a block-uniform algorithm can obtain a uniformly random $r$-square.

Recall the definition of the border of an image from Section~\ref{sec:defintions_notation}.
\begin{definition}[Border connectedness]\label{def:border_connectedness}
A (sub)image $S$ is {\em border-connected} if for every black pixel $(i,j)$ of $S$, the image graph $G_S$ contains a path from $(i,j)$ to a pixel on the border. The property
\emph{border connectedness,} denoted $\C'$, is the set of all border-connected images.
\end{definition}


\ifnum\full=1
\subsection{Proof of Theorem~\ref{thm:connectedness_dist_appr}}\label{sec:proof_of_thm}
\fi
The main idea behind Algorithm~\ref{alg:connectedness}, used to prove Theorem~\ref{thm:connectedness_dist_appr}, is to relate the distance to connectedness to the distance to another property, which we call {\em grid connectedness}. The latter distance is the average over squares of the distances of these squares to {\em border connectedness}. The average can be easily estimated by looking at a sample of the squares.


W.l.o.g.\ assume that $n \equiv 1 \pmod {4/\mydelta}$. (Otherwise, we can pad the image with white pixels without changing whether it is connected and adjust the accuracy parameter.)

\begin{algorithm}\caption{Distance approximation to connectedness.}
\SetKwInOut{Input}{input}\SetKwInOut{Output}{output}
\label{alg:connectedness}
\Input{$n\in\mathbb{N}$ and $\mydelta\in(0,1/4)$; block-sample access to an $n\times n$ binary matrix $M$.}
\DontPrintSemicolon
\BlankLine
\nl Sample $s= 4/\mydelta^2$ squares uniformly and independently from $S_{4/\mydelta}$ (see Definition~\ref{def:squares}).\;
\tcp{This can be done by drawing random blocks from the $4/\mydelta$-partition of $\domain$.}
\nl\label{step:ave-dist-border-conn} For each such square $S$,
compute $\dis(S,\C')$ (see \ifnum\full=1 Section~\ref{sec:border-con}\else the {\color{black} full version}\fi \ for details), where $\C'$ is border connectedness (see Definition~\ref{def:border_connectedness}). Let $\dhatsquares$ be the average of computed distances $\dis(S,\C')$.\;
\nl \Return $\dout = \left((1-\frac \mydelta 4) (1-\frac 1n)\right)^2\cdot\dhatsquares$.
\end{algorithm}

\begin{definition}\label{def:grid_pixels} Fix $\mydelta\in(0,1/4)$. Let image $M_\mydelta$ be a {\em gridded image} obtained from image $M$ as follows:
\ifnum\full=0
\begin{center}
$M_\mydelta[i,j]=
\begin{cases}
1 &\text{ if } (i,j)$ is a grid pixel from $\gp_{4/\mydelta};\\
M[i,j] & \text{ otherwise.}
\end{cases}
$
\end{center}
\else
$$M_\mydelta[i,j]=
\begin{cases}
1 &\text{ if } (i,j)$ is a grid pixel from $\gp_{4/\mydelta};\\
M[i,j] & \text{ otherwise.}
\end{cases}
$$
\fi
Let $\C$ be the set of all connected images. For $\mydelta\in(0,1/4),$ define \emph{grid connectedness}
$\C_{\mydelta}=\{ M \mid  M \in\C  \text{, and $M[i,j]=1$ for all $(i,j)\in \gp_{4/\mydelta}$}\}.$
\end{definition}
\begin{lemma}\label{lem:distance_relation}
 Let $d_M=\dis(M,\C)$ and  $d_\mydelta=\dis(M_\mydelta,\C_{\mydelta})$. Then $d_M-\frac{\mydelta}{2}\leq d_\mydelta\leq d_M$. Moreover,
$$d_\mydelta=\Bigl(\bigl(1-\frac \mydelta 4\bigr) \bigl(1-\frac 1n\bigr)\Bigr)^2\cdot\frac 1 {|S_{4/\mydelta}|}\sum_{S\in S_{4/\mydelta}}\dis(S,\C').$$
\end{lemma}
\begin{proof}
First, we prove that $d_\mydelta\leq d_M.$
Let $M'$ be a connected image such that $\dis(M,M')=d_M$. Then $M'_\mydelta$, the gridded image obtained from $M'$, satisfies $\C_{\mydelta}$. Since $\dis(M_\mydelta,M'_\mydelta)\leq d_M$, it follows that $d_\mydelta\leq d_M.$

Now we show that $d_M-\frac{\mydelta}{2}\leq d_\mydelta$. Let $M''_\mydelta\in \C_\mydelta$ be such that $\dis(M_\mydelta,M''_\mydelta)=d_\mydelta$. Then $M''_\mydelta\in \C$ and, by Claim~\ref{claim:GP-size}, $\dis(M,M''_\mydelta)\leq |\gp_{4/\mydelta}|/n^2 +d_\mydelta\leq \mydelta/2+d_\mydelta$, implying $ d_M\leq \mydelta/2+d_\mydelta$, as required.

Finally, observe that to make $M_\mydelta$ satisfy $\C_\mydelta$, it is necessary and sufficient to ensure that each square satisfies $\C'$. In other words,
$$d_\mydelta n^2=\sum_{S\in S_{4/\mydelta}}\Dis(S,\C')= (4/\mydelta-1)^2\sum_{S\in S_{4/\mydelta}}\dis(S,\C').$$
Since $|S_{4/\mydelta}|=(\frac{n-1}{4/\mydelta})^2$, the desired expression for $d_\mydelta$ follows.
\end{proof}

\ifnum\full=1

\subparagraph{Analysis of Algorithm~\ref{alg:connectedness}.}
Let $\dsquares =\frac 1 {|S_{4/\mydelta}|}\sum_{S\in S_{4/\mydelta}}\dis(S,\C')$. Recall that $\dhatsquares$ is the empirical average computed by the algorithm.
By the Chernoff-Hoeffding bound, $\Pr[|\dhatsquares-\dsquares|>\mydelta/2]\leq 2\exp(-2\mydelta^{2}s)\leq 1/3$. So, with probability at least 2/3, we have $|\dhatsquares-\dsquares|\leq \mydelta/2$. If this event happens then
$|\dout-d_\mydelta|\leq \mydelta/2$ because by Algorithm~\ref{alg:connectedness} and Lemma~\ref{lem:distance_relation}, respectively, $\dout=A\cdot\dhatsquares$ and $d_\mydelta= A\cdot\dsquares$, where $A=\left((1-\frac \mydelta 4) (1-\frac 1n)\right)^2 \leq 1$. By Lemma~\ref{lem:distance_relation}, $|d_\mydelta-\eps|\leq \mydelta/2$. Thus,
$|\dout-\eps|
\leq |\dout-d_\mydelta|+ |d_\mydelta-\eps|
\leq \mydelta/2+\mydelta/2=\mydelta$ holds with probability at least 2/3, as required.

\subparagraph{Query and time complexity.}
Algorithm~\ref{alg:connectedness} samples $O(1/\mydelta^2)$ squares containing $O(1/\mydelta^2)$ pixels each. Thus, the sample complexity is $O(1/\mydelta^4)$.

The most expensive step in Algorithm~\ref{alg:connectedness} is Step~\ref{step:ave-dist-border-conn} where the distance of a square $S$ to border connectedness is calculated. By Theorem~\ref{thm:border-con} (see section \ref{sec:border-con}), the running time of this step for one square is $\exp\left(O\left(\frac 1 \mydelta \right)\right)$ and it is called $O(1/\eps^2)$ times. Therefore, the running time of Algorithm~\ref{alg:connectedness} is $\exp\left(O\left(\frac 1 \mydelta \right)\right)$, as claimed.
\else
The rest of the analysis is completed in the {\color{black} full version of this paper} using the Chernoff-Hoeffding bound.
\fi

\subsection{Algorithm for Border Connectedness}
\label{sec:border-con}

\begin{theorem}
\label{thm:border-con}
Let $S$ be a $k\times k$ image. There is an algorithm that computes $\dis(S,\C')$ (i.e., distance of $S$ to border connectedness) in time $\exp \left (O(k)\right)$.
\end{theorem}

\begin{proof}{\color{black}
To prove the theorem we give a dynamic programming algorithm that computes $\dis(S,\C')$ in the following way: starting from row 1 of $S$, it processes a row and proceeds to the next one. The algorithm stops after processing row $k$. The information the algorithm computes and stores for each row is explained later in this section.

\begin{definition}
\label{def:block}
Let $\overline{cl}\in\{0,1\}^k$ be a vector. Call maximal consecutive runs of $1$'s in $\overline{cl}$ {\em 1-blocks} and let $n(\overline{cl})$ denote the number of 1-blocks in $\overline{cl}$. Let $\bold{1}^t$ (respectively, $\bold{0}^t$) denote the string of $t$ ones (respectively, zeros) and let $\Sigma=\{0,1,<,\times,>\}$.
\end{definition}
}

Consider a $k\times k$ image $S$. Recall that $G_{S}$ denotes the image graph of $S$. For every $i\in[k]$, denote the subgraph of $G_{S}$, induced by the first $i$ rows in $S$, by $G^i_{S}$. Index 1-blocks in row $i$ of $S$ in the increasing order of indices of pixels they contain. For example, a row $001110011$ contains two 1-blocks; the 1-block with three $1$'s  has index 1, the 1-block with two $1$'s has index 2. Each 1-block in row $i$ has one of the following 5 statuses w.r.t.\ $G^i_{S}$:
\begin{itemize}
\item{connected to the border of $S$ (denoted by 1);}
\item{isolated, i.e., not connected to the border and to any other 1-block in its row (denoted by 0);}
\item{first 1-block in its connected component, i.e., it is in the connected component with other 1-blocks of row $i$ and has the smallest index among them (denoted by $<$);}
\item{intermediate 1-block in its connected component, i.e., has neither largest nor smallest index in its connected component (denoted by $\times$);}
\item{last 1-block in its connected component, i.e., it is in the connected component with other 1-blocks of row $i$ and has the largest index among them (denoted by $>$).}
\end{itemize}

{\color{black}
\begin{definition}\label{def:config}
Let $\overline{cl}$ denote the coloring of $S$ in row $i$, for some $i\in[k]$ (i.e., $\overline{cl}=S[i]$). Statuses of 1-blocks of $\overline{cl}$ w.r.t. $G^i_S$ are captured by a {\em status vector} $\overline{st}\in\Sigma^{n(\overline{cl})}$. The pair $(\overline{cl},\overline{st})$ is called the {\em configuration w.r.t.\ $G^i_{S}$.}
\end{definition}

\begin{definition}\label{def:bc-sets}
For all $i\in[k]$, $\overline{cl}\in\{0,1\}^k$, and vectors $\overline{st}$ over $\Sigma$ of length at most $k$, define $B(i,\overline{cl},\overline{st})=\{S'\mid S' \text{ is a }k\times k \text{ border-connected image with configuration }(\overline{cl},\overline{st}) \text{ w.r.t.\ } G^i_{S'}\}.$
\end{definition}

For all $i\in[k]$ and $k\times k$ images $S'$, let $cost_i(S,S')$ denote the number of pixels on which the first $i$ rows of $S$ and $S'$ differ. For all $i\in[k], \overline{cl}\in\{0,1\}^k$, and $\overline{st}\in\Sigma^k$ define
$$
{cost}(i,\overline{col},\overline{st})=
\begin{cases}
\min_{S'\in B(i,\overline{cl},\overline{st})}(cost_i(S,S')) & \text{ if } B(i,\overline{col},\overline{st})\neq\emptyset,\\
\infty & \text{ otherwise.}
\end{cases}
$$
Note that the number of all possible configurations for a row is at most $2^k\cdot 5^k=\exp(k)$. We show that if for some $i\in[k-1]$, the value of ${cost}(i,\overline{cl},\overline{st})$ is known for every configuration $(\overline{cl},\overline{st})$ in row $i$, then for every configuration $(\overline{cl'},\overline{st'})$ in row $i+1$, the cost ${cost}(i+1,\overline{cl'},\overline{st'})$ can be computed in time exponential in $k$. This is a crucial ingredient that helps us to show that the running time of our algorithm is exponential in $k$.

For a fixed $i\in[k-1]$, consider an image $S'\in B(i,\overline{cl},\overline{st})$. Let $S''$ be an image which has the same $j\in[i]$ rows as $S'$. Let $\overline{cl'}$ denote the coloring of row $i+1$ in $S''$. If $S''$ is border-connected then configuration $(\overline{cl},\overline{st})$ is {\em consistent} with coloring $\overline{cl'}$, i.e., every 1-block in $\overline{cl}$ that has status other than 1 w.r.t.\ $G^{i}_{S''}$ is connected to a 1-block in $\overline{cl'}$ w.r.t.\ $G^{i+1}_{S''}$. Moreover, for some status vector $\overline{st'}\in\Sigma^{n(\overline{cl'})}$, image $S''\in B(i+1,\overline{cl'},\overline{st'})$ and $\overline{st'}$ can be determined from $\overline{cl},\overline{cl'}$, and $\overline{st}$. Observe that if ${cost}(i,\overline{cl},\overline{st})$ is known for every configuration $(\overline{cl},\overline{st})$, then $cost(i+1,\overline{cl'},\overline{st'})$ can be computed. After computing costs for all configurations in each row, $\dis(S,\C')=\min_{\overline{cl},\overline{st}}cost(k,\overline{cl},\overline{st})$ can be found. In order to find $\overline{st'}$, our algorithm uses subroutine \Compst. \Compst uses subroutine \Constgr that creates a graph whose nodes are 1-blocks of $\overline{cl}$ and edges are defined according to the information provided by $\overline{st}$. Now we explain how $\overline{st'}$ is found. }

Now we show how subroutine \Compst computes $\overline{st'}$ from colorings $\overline{col}$,$\overline{col'}$, and the status vector $\overline{st}\in\Sigma^{num(\overline{col})}$ of $\overline{col}$ w.r.t\ $G^i_{I}$, where $I$ is an image from $B(i,\overline{col},\overline{st})$ for $i\in[k-1]$. Let $n_1=num(\overline{col})$ and $n_2=num(\overline{col'})$. Index 1-blocks in $\overline{col}$ in the nondecreasing order of indices of pixels they contain. Let $I'$ be an image obtained from $I$ by recoloring its row $i+1$ to $\overline{col'}$. Index 1-blocks in $\overline{col'}$ in the nondecreasing order of indices of pixels they contain and add $n_1$ to each index. Consider graph $G=(V,E)$ where $V=[n_1+n_2]$ and $E$ has every edge of the following two types:
\begin{enumerate}
\item edges $(i,j)$, where $i,j\in[n_1]$, $i<j$, and $i$ is not connected to any $j'<j$ in  $G^{i}_{I'}$.
\item $(i,n_1+j)$, where $i\in[n_1], j\in[n_2]$, and $i$ is connected to $n_1+j$ in $G^{i+1}_{I'}$.
\end{enumerate}

\Compst uses subroutine \Constgr to construct graph $G$. (\Constgr computes set $E$.)
After graph $G$ is constructed, \Compst checks whether configuration  $(\overline{col},\overline{st})$ and $\overline{col'}$ are consistent w.r.t.\ $G^{i+1}_{I'}$. If they are not consistent it outputs a $\perp$ symbol. If they are consistent then $B(i+1,\overline{col'},\overline{st'})\neq\emptyset$ (rows $i+1,\ldots, k$ in $I'$ can be recolored to all black rows and the resulting image is in $B(i+1,\overline{col'},\overline{st'})$). \Compst finds vector $\overline{st'}$ based on the connectivity information of graph $G$ and information provided by vector $\overline{st}$.

To check whether $\overline{cl'}$ is consistent with $(\overline{col},\overline{st})$ our algorithm uses subroutine \Compst (Algorithm~\ref{alg:status}). If it is consistent, to
find the status vector $\overline{st'}$. Subroutine \Compst uses subroutine \Constgr that constructs a graph from $\overline{cl}$ and $\overline{st}$ that helps to compute $\overline{st'}$. The nodes in this graph are 1-blocks of $\overline{cl}$ and edges are defined according to the information provided by $\overline{st}$. \Constgr is explained next.

For each $i\in[k]$, every image $S$ has some configuration $(\overline{cl},\overline{st})$ in its $i$'th row w.r.t.\ $G^i_S$. Vector $\overline{st}$ in this  configuration has the information about which 1-blocks in $\overline{cl}$ are in the same connected component in $G^i_S$ and which are connected to the border. Thus, if $(\overline{cl},\overline{st})$ is given we can construct a graph $G$ whose nodes are all 1-blocks of $\overline{cl}$ and the status vector of $\overline{cl}$ w.r.t.\ $G$ is $\overline{st}$. Subroutine \Constgr (Algorithm~\ref{alg:graph}) constructs graph $G$.

\begin{algorithm}\label{alg:graph}
\caption{Subroutine \Constgr used in Algorithm~\ref{alg:status}.}
\SetKwInOut{Input}{input}\SetKwInOut{Output}{output}
\Input{vector $\overline{cl}\in\{0,1\}^k$, and $\overline{st}\in\Sigma^{n(\overline{cl})}$.}
\DontPrintSemicolon
\BlankLine
\nl Index 1-blocks in $\overline{cl}$ in the nondecreasing order of indices of pixels they contain. \\
\tcp{Let $n_1=num(\overline{cl}), V=[n_1], E=\emptyset,$ and $stack=\emptyset$ (we maintain a stack)}

\nl \ForAll {indices $j=1,2,\ldots,n_1$} {
\nl\textbf{if} $stack\neq\emptyset$ and $\overline{st}[j]=1$ \textbf{then} \Return $\perp$\\
\nl\textbf{if} $stack=\emptyset$ and $\overline{st}[j]\in\{\boldsymbol{\times},\boldsymbol{>}\}$ \textbf{then} \Return $\perp$\\

\nl\textbf{if} $\overline{st}[j]\in\{\boldsymbol{<},\boldsymbol{\times}\}$ \textbf{then} $push(j)$\\
\nl\textbf{if} $\overline{st}[j]=\ \boldsymbol{>}$ \textbf{then} \textbf{do} $p=pop(stack)$; add $\{p,j\}$ to $E$ \textbf{until} $p=\ \boldsymbol{<}$\\
}
\nl \Return $G=(V,E)$\\
\end{algorithm}

\begin{algorithm}
\caption{Distance to border connectedness of a square $S$.}
\label{alg:dist}
\SetKwInOut{Input}{input}\SetKwInOut{Output}{output}
\Input{access to a $k\times k$ square $S$.}
\DontPrintSemicolon
\BlankLine
\nl\label{st:init}\ForAll  {indices $i\in[k]$, vectors $\overline{col}\in\{0,1\}^k$, and $\overline{st}\in \Sigma^{num(\overline{col})}$} \do{
\tcp{For $i\in[k]$, let $\overline{r}_{i}\in\{0,1\}^{k}$ be a vector that corresponds to the $i^{th}$ row of $S$.}
\nl\quad$cost(i,\overline{col}, \overline{st})=
\begin{cases} |\overline{r}_{1}-\overline{col}|_1, & \mbox{if } i=1\mbox{ and }\overline{st}=\bold{1}^{num(\overline{col})};\\
\infty, & \mbox{otherwise.}\end{cases}$\\}

\nl\label{st:subsequent}\ForAll {indices $i=2,3,...,k$, vectors $\overline{col},\overline{col'}\in\{0,1\}^k$ and $\overline{st}\in\Sigma^{num(\overline{col})}$}
   \do{
\nl\quad \textbf{if} $cost(i-1,\overline{col},\overline{st})\neq\infty$ \textbf{then} $\overline{st'}=\Compst(i,\overline{col},\overline{st},\overline{col'})$\\
\nl\quad\textbf{if} $\overline{st'}\neq\perp$ \textbf{then} $cost(i,\overline{col'},\overline{st'})=\min\{cost(i,\overline{col'},\overline{st'}), cost(i-1,\overline{col},\overline{st})+|\overline{r}_{i}-\overline{col'}|_1\}$

}
\nl \label{st:actual-min}\Return $(\min_{\overline{col}\in\{0,1\}^k, \overline{st}\in\Sigma^{num(\overline{col})}}cost(k,\overline{col}, \overline{st}))\cdot k^{-2}$
\end{algorithm}

\subparagraph{Analysis of Algorithm~\ref{alg:dist}.}
The following lemma shows that Algorithm~\ref{alg:dist} is correct.
\begin{lemma}
\label{lm:correct}
For all $i\in[k]$, $\overline{col}\in\{0,1\}^k$, and $\overline{st}\in\Sigma^{k}$, Algorithm~\ref{alg:dist} correctly computes ${cost}(i,\overline{col}, \overline{st})$.
\end{lemma}
\begin{proof} We prove the lemma inductively. For the first row, Algorithm~\ref{alg:dist} indeed computes the cost of every configuration. (Note that every 1-block in the first row is connected to the border and thus, every such 1-block has status $1$. Therefore, $cost(1,\overline{col},\overline{st})=|r_1-\overline{col}|_1$ if $\overline{st}=1^{num(\overline{col})}$, and $cost(1,\overline{col},\overline{st})=\infty$, otherwise). Assume that the statement in the lemma holds for some row $i\in[k-1]$. We prove the statement for row $i+1$. Note that if $B(i+1,\overline{col'},\overline{st'})=\emptyset$ for some $\overline{col'},\overline{st'}$, then $cost(i+1,\overline{col'},\overline{st'})=\infty$. The algorithm correctly sets $cost(i+1,\overline{col'},\overline{st'})$ to $\infty$ and never changes it. If $B(i+1,\overline{col'},\overline{st'})\neq\emptyset$ consider an image $I^*\in B(i+1,\overline{col'},\overline{st'})$ such that $cost(i+1,\overline{col'},\overline{st'})=cost_{i+1}(S,I^*)$. Let $\overline{col}$ be the coloring of row $i$ in $I^*$ and $\overline{st}$ be the status vector of $\overline{col}$ w.r.t.\ $G^i_{I^*}$. Then $\overline{col'}$ and the configuration $(\overline{col},\overline{st})$
are consistent w.r.t. $G^{i+1}_{I^*}$. Note that $I^*\in B(i,\overline{col},\overline{st})$ and $cost_{i+1}(S,I^*)=cost_i(S,I^*)+|r_{i+1}-\overline{col'}|_1$. Moreover, for every image $I_1\in B(i,\overline{col},\overline{st})$, there is an image $I_2\in B(i+1,\overline{col'},\overline{st'})$ such that $cost_i(S,I_1)=cost_i(S,I_2)$. (In $I_2$, recolor row $i+1$ to $\overline{col'}$ and all rows $i+2,\ldots, k$ to all black rows and obtain image $I_2$. In image $I_2$, $\overline{col'}$ and $(\overline{col},\overline{st})$ are consistent w.r.t.\ $G^{i+1}_{I_2}$ and every 1-block in its rows $i+1,\ldots,k$ is border-connected.  Thus, image $I_2$ is border-connected and $I_2\in B(i+1,\overline{col'},\overline{st'})$.) Therefore, $cost_i(S,I^*)=\min_{I\in B(i,\overline{col},\overline{st})}cost_i(S,I)=cost(i,\overline{col},\overline{st})$. At some point, the algorithm considers the configuration $(\overline{col},\overline{st})$ for row $i$ and the coloring $\overline{col'}$ for row $i+1$. By the inductive assumption, the algorithm correctly computes $cost(i,\overline{col},\overline{st})$ which is equal to $cost_i(S,I^*)$. The output of \Compst for the triple $i,\overline{col},\overline{st}$ will be the vector $\overline{st'}$ and the algorithm sets $cost(i+1,\overline{col'},\overline{st'})$ to $cost(i,\overline{col},\overline{st})+|r_{i+1}-\overline{col'}|_1=cost_i(S,I^*)+|r_{i+1}-\overline{col'}|_1$ which is the correct value. This completes the proof.
\end{proof}

By Lemma~\ref{lm:correct}, Algorithm~\ref{alg:dist} computes the cost of every configuration in row $k$. The algorithm outputs the minimum one among these costs. Let $\hat{S}$ be an image such that $\dis(S,\C')=\dis(S,\hat{S})$. Note that configurations of row $k$ that are not possible for a border-connected image have unbounded costs (i.e., $\infty$). Thus, row $k$ in $\hat{S}$ has some configuration which has the minimum cost among all configuration costs for the row. Note that the cost of a configuration in row $k$ is equal to the cost of recoloring of $S$ to some border connected square. Therefore, the output of Algorithm~\ref{alg:dist} is equal to $\dis(S,\C')$.

\begin{algorithm}\label{alg:status}
\caption{Subroutine \Compst used in Algorithm~\ref{alg:dist}.}
\SetKwInOut{Input}{input}\SetKwInOut{Output}{output}
\Input{index $i$; vectors $\overline{col},\overline{col'}\in\{0,1\}^k$, and $\overline{st}\in\Sigma^{num(\overline{col})}$.}
\DontPrintSemicolon
\BlankLine
\nl Construct graph $G=\Constgr(\overline{col},\overline{col'},\overline{st})$\tcp{Let $G=(V,E)$.}
\nl \textbf{if} $E=\emptyset$ \textbf{then} \Return $\perp$\\
\tcp{Let $n_1=num(\overline{col}), n_2=num(\overline{col'})$}
\nl Let $\overline{col}[0]=\overline{col'}[0]=row_1=row_2=0$\\
\nl For every $j=1,2,\ldots,n_1$\\
\nl\quad \textbf{if} $\overline{col}[j-1]=0$ and $\overline{col}[j]=1$ \textbf{then} increment $row_1$ by 1\\
\nl\quad \textbf{if} $\overline{col'}[j-1]=0$ and $\overline{col'}[j]=1$ \textbf{then} increment $row_2$ by 1\\
\nl\quad \textbf{if} $\overline{col}[j]\cdot \overline{col'}[j]=1$ and $\{row_1,row_2\}\notin E$ \textbf{then} add $\{row_1,row_2\}$ to $E$\\

\nl \textbf{if} $\exists j\in[n_1]$ with $\overline{st}[j]\neq 1$ such that $(j,n_1+j')\notin E$ for all $j'\in[n_2]$ \textbf{then} \Return $\perp$

\nl \textbf{if} $i=k$ \textbf{then} $\overline{st'}=\bold{1}^{n_2}$\\
\nl \textbf{else}\\ 
\nl\quad Let $\overline{st'}=\bold{0}^{n_2}$. Update $\overline{st'}[1]=\overline{col'}[1]$ and $\overline{st'}[n_2]=\overline{col'}[k]$

\nl\quad For each edge $(j,n_1+j')\in E$, $j\in[n_1],j'\in[n_2]$, \textbf{if} $\overline{st}[j]=1$ \textbf{then} $\overline{st}[n_1+j']=1$.

\nl\quad {Run BFS to find connected components in $G$. For each pair $(n_1+j,n_1+j')$, where \\ \quad $j,j'\in[n_2]$, \textbf{if} $j$ and $j'$ are connected and $\overline{st'}[j]=1$ \textbf{then} $\overline{st'}[n_1+j']=1$.}

\nl\quad{For each connected component of vertices $n_1+j$ not marked by 1, where $j\in[n_2]$, update \\ \quad the corresponding entries of $\overline{st'}$ with the corresponding symbols in $\Sigma$ (i.e., $\overline{st'}[j]=\ \boldsymbol{<}$ if it is\\ \quad the vertex with the smallest index in the component, $\overline{st'}[j]=\ \boldsymbol{>}$ if it is the vertex with the \\ \quad largest index in the component, and $\overline{st'}[j]= \boldsymbol{\times}$, otherwise).}

\nl \Return $\overline{st'}$

\end{algorithm}

\subparagraph{Query and Time Complexity of Algorithm~\ref{alg:dist}.}
The most expensive step in Algorithm~\ref{alg:dist} is Step 3. Note that there are at most $k\cdot2^k\cdot5^k$ sets $B(\cdot,\cdot,\cdot)$ for which subroutine \Compst is called in this step. \Compst uses subroutine \Constgr to construct graph $G=(V,E)$. To construct type 1 and type 2 edges in $E$ subroutine \Constgr performs $O(n_1)+O(n_1)=O(k)$ operations. Thus, the running time of \Constgr is $O(k)$ (recall that $n_1=num(\overline{col})$, $n_2=num(\overline{col'})$). Therefore, Steps 1-5 of \Compst run in time $O(k)$. Among the remaining steps (Steps 6-10) of \Compst the most expensive ones are Steps 8 and 9. Each of them runs in time $O(n_1\cdot n_2)=O(k^2)$ time. Thus, the running time of \Compst is $O(k^2)$ and Algorithm~\ref{alg:dist} runs in time $O( k^2\cdot k2^k5^k)=\exp \left (O(k)\right)$, as claimed.
\end{proof}

\bibliographystyle{abbrvnat}
\bibliography{visual-properties}

\end{document}